\documentclass[letterpaper,11pt]{article}

\usepackage[utf8]{inputenc}
\usepackage{graphicx}
\usepackage{fullpage}
\usepackage{xspace}
\usepackage[margin=1in]{geometry}
\usepackage{subfig}
\usepackage{tikz}
\usetikzlibrary{fit}
\usepackage{times}
\usepackage{float,verbatim,amsthm,amssymb,amsmath}
\newtheorem{definition}{Definition}
\newtheorem{lemma}{Lemma}

\newtheorem{theorem}{Theorem}

\tikzset{defnode/.style={draw, circle, fill=gray, inner sep=.6pt}}
\tikzset{mis/.style={draw, circle, red, inner sep=1.5pt}}
\xdefinecolor{red}{rgb}{0.8,0.1,0.1}
\xdefinecolor{green}{rgb}{0.0,0.7,0.0}

\floatstyle{ruled}
\newfloat{algorithm}{tp}{lop}
\floatname{algorithm}{Algorithm}

\title{Robustness in Highly Dynamic Networks}

\author{Arnaud Casteigts$^1$ \and Swan Dubois$^2$ \and Franck Petit$^2$ \and John Michael Robson$^1$}

\date{}

\begin{document}

\newcommand{\forallMIS}{\ensuremath{{\cal RMIS^\forall}}\xspace}
\newcommand{\existsMIS}{\ensuremath{{\cal RMIS^\exists}}\xspace}
\newcommand{\RMIS}{{\sc RobustMIS}\xspace}
\newcommand{\TC}{{\ensuremath{{\cal TC^{\cal R}}}}\xspace}
\newcommand{\vs}{{\em vs.}\xspace}
\newcommand{\ie}{{i.e.}\xspace}
\newcommand{\eg}{{e.g.}\xspace}
\newcommand{\etc}{{\em etc.}}
\newcommand{\aka}{{\em a.k.a.}\xspace}
\newcommand{\fixme}[1]{\fbox{\textsl{{\bf #1}}}}
\newcommand{\FIXME}[1]{\fixme{#1} \marginpar[\null\hspace{2cm} FIXME]{FIXME}}
\newcommand{\wrt}{{\em w.r.t.}\xspace}

\newcommand{\typePI}{\fbox{\footnotesize{\textsl{PI}}}\xspace}
\newcommand{\typePIdesc}{Possibly In\xspace}

\newcommand{\typePO}{\fbox{\footnotesize{\textsl{PO}}}\xspace}
\newcommand{\typePOdesc}{Possibly Out\xspace}

\newcommand{\typePE}{\fbox{\footnotesize{\textsl{PE}}}\xspace}
\newcommand{\typePEdesc}{Possibly External\xspace}

\newcommand{\typeN}{\fbox{\footnotesize{\textsl{N}}}\xspace}
\newcommand{\typeNdesc}{Negative\xspace}

\newcommand{\typeE}{\fbox{\footnotesize{\textsl{E}}}\xspace}
\newcommand{\typeEdesc}{End\xspace}

\newcommand{\ABCT}{\ensuremath{\mathcal{ABC}}-tree\xspace}
\newcommand{\ABCST}{\ensuremath{\mathcal{ABC}}-subtree\xspace}
\newcommand{\A}{\ensuremath{\mathcal{A}}\xspace}
\newcommand{\B}{\ensuremath{\mathcal{B}}\xspace}
\newcommand{\C}{\ensuremath{\mathcal{C}}\xspace}
\newcommand{\PV}{\ensuremath{\mathcal{P}}\xspace}
\newcommand{\BS}{\ensuremath{\mathcal{S}}\xspace}

\newcommand{\CT}{\ensuremath{{\cal T}_G}\xspace}
\newcommand{\R}{\ensuremath{\mathcal{R_{\CT}}}\xspace}
\newcommand{\LONG}[1]{#1}

\newenvironment{missingproof}{}{{\hfill $\Box$}\vspace{.5pc}}

\setlength\textfloatsep{10pt plus 2pt minus 2pt}

\maketitle

\thispagestyle{empty}

\begin{abstract}
  We investigate a special case of hereditary property that we refer to as {\em robustness}. A property is {\em robust} in a given graph if it is inherited by all connected spanning subgraphs of this graph. We motivate this definition in different contexts, showing that it plays a central role in
highly dynamic networks, although the problem is defined in terms of classical (static) graph theory. 
  In this paper, we focus on the robustness of {\em maximal independent sets} (MIS). Following the above definition, a MIS is said to be {\em robust} (RMIS) if it remains a valid MIS in all connected spanning subgraphs of the original graph.
We characterize the class of graphs in which {\em all} possible MISs are robust. We show that, in these particular graphs, the problem of
  finding a robust MIS is {\em local}; that is, we present an RMIS algorithm using only
  a sublogarithmic number of rounds (in the number of nodes $n$) in the ${\cal LOCAL}$ model. On the negative side, we show that, in general graphs, the problem is not local. Precisely, we prove a $\Omega(n)$ lower bound on the number of rounds required for the nodes to decide consistently in some graphs. This result implies a separation between the RMIS problem and the MIS problem in general graphs.  It also implies that any strategy in this case is asymptotically (in order) as bad as
  collecting all the network information at one node and solving the problem in a centralized manner. Motivated by this observation, we present a centralized algorithm that computes a robust MIS in a given graph, if one exists, and rejects otherwise. Significantly, this algorithm requires only a polynomial amount of local computation time, despite the fact that exponentially many MISs and exponentially many connected spanning subgraphs may exist.
\end{abstract}

\footnotetext[1]{Université de Bordeaux, CNRS, LaBRI UMR 5800, France}
\footnotetext[2]{UPMC Sorbonne Universit\'es, CNRS, Inria, LIP6 UMR 7606, France}

\setcounter{page}{1}

\section{Introduction}
\label{sec:introduction}

Highly dynamic networks are made of dynamic (often mobile) entities such as vehicles, drones, or robots. 
It is generally assumed, in these networks, that the set of entities (nodes) is constant, while the set of communication
links varies over time. Many classical assumptions do not hold in these networks. For example, the topology may be
disconnected at any instant. It may also happen that an edge present at some time never appears again in the future.
In fact, of all the edges that appear at least once, one can distinguish between two essential sets: the set of {\em
recurrent} edges, which always reappear in the future (or remain present), and the set of {\em non recurrent} edges
which eventually disappear in the future. The static graph containing the union of both edge sets is called the {\em
footprint} of the network~\cite{CFQS12}, and its restriction to the recurrent edges is the {\em eventual footprint} 
of the network~\cite{DLLP16}.

It is not clear, at first, what assumptions seem reasonable in a highly dynamic network. Special cases have been considered recently, such as 
always-connected dynamic 
networks~\cite{OW05}, $T$-interval connected networks~\cite{KLO10}, or networks the edges of which correspond to pairwise interactions obeying a uniform random scheduler (see e.g.~\cite{AADFP06,MCS11}). 
Arguably, one of the weakest possible assumption is that any pair of nodes be able to communicate infinitely often through {\em temporal paths} (or journeys). Interestingly enough, this property was identified more than three decades ago by Awerbuch and Even~\cite{AE84} and remained essentially ignored afterwards. The corresponding class of dynamic networks (Class~5 in~\cite{CFQS12}---here referred to as \TC for consistency with various notations~\cite{DKP15,GCLL15,AGMS15}) is however one of the most general and it actually includes the three aforementioned cases. 

Dubois {\it et al.}~\cite{DKP15} observe that class $\TC$ is actually the set of dynamic networks whose {\em eventual footprint} is connected. In other words, it is more than reasonable to assume that some of the edges are recurrent and their union does form a {\em connected} spanning subgraph.
Solving classical problems such as symmetry-breaking tasks relative to this particular set thus makes sense, as the nodes can rely forever on the corresponding solution, even though intermittently~\cite{CF13b,DKP15}. 
Unfortunately, it is impossible for a node to distinguish between the set of recurrent edges and the set of non recurrent edges.  So, the best the nodes can do is to compute a solution relative to the footprint, hoping that this solution still makes sense in the eventual footprint, whatever it is. (Whether, and how the nodes can learn the footprint itself is discussed later on.)

This context suggests a particular form of heredity which we call {\em robustness}. In classical terms, robustness can be formulated as the fact that a given property must be inherited by all the connected spanning subgraphs of the original graph. Significantly, this concept admits several possible interpretations, including the dynamic interpretation developed here. A more conventional, almost direct interpretation is that some edges in a classical (static) network are subject to permanent failure at some point, and the network is to be operated so long as it remains connected. While this interpretation is more intuitive and familiar, we insist on the fact that the dynamic interpretation of robustness is what makes its study compelling, for this notion arises naturally in class \TC, which is one of the most general class of dynamic networks imaginable. The reader may adopt either interpretation while going through the paper, keeping in mind that our results apply to both contexts and are therefore quite general.

\noindent\textbf{Contributions.}
We investigate the concept of {\em robustness} of a property, with a focus on the {\em maximal independent set} (MIS) problem, which consists of selecting a subset of nodes none of which are neighbors (independence) and such that no further node can be added to it (maximality). 
As it turns out, a robust MIS may or may not exist, and if it exists, it may or may not be computable locally depending on the considered graph ({\it resp.} footprint). 
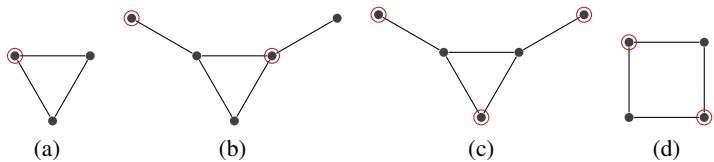
\begin{figure}[h]
  \centering
  \subfloat[]{
    \label{fig:mis-a}
    \begin{tikzpicture}
      \tikzstyle{every node}=[circle, inner sep=1.2pt, fill=darkgray]
      \path (0,0) node (a) {};
      \path (a)+(0:1) node (b) {};
      \path (a)+(-60:1) node (c) {};
      
      \draw (a)--(b)--(c)--(a);
      
      \tikzstyle{every node}=[circle, draw, red, inner sep=2pt]
      \path (a) node {};
    \end{tikzpicture}
  }
  ~
  \subfloat[]{
    \label{fig:mis-b}
    \begin{tikzpicture}
      \tikzstyle{every node}=[circle, inner sep=1.2pt, fill=darkgray]
      \path (0,0) node (a) {};
      \path (a)+(-30:1) node (b) {};
      \path (b)+(0:1) node (c) {};
      \path (b)+(-60:1) node (d) {};
      \path (c)+(30:1) node (e) {};
      
      \draw (a)--(b)--(c)--(e);
      \draw (b)--(d)--(c);
      
      \tikzstyle{every node}=[circle, draw, red, inner sep=2pt]
      \path (a) node {};
      \path (c) node {};
    \end{tikzpicture}
  }
  ~
  \subfloat[]{
    \label{fig:mis-c}
    \begin{tikzpicture}
      \tikzstyle{every node}=[circle, inner sep=1.2pt, fill=darkgray]
      \path (0,0) node (a) {};
      \path (a)+(-30:1) node (b) {};
      \path (b)+(0:1) node (c) {};
      \path (b)+(-60:1) node (d) {};
      \path (c)+(30:1) node (e) {};
      
      \draw (a)--(b)--(c)--(e);
      \draw (b)--(d)--(c);
      
      \tikzstyle{every node}=[circle, draw, red, inner sep=2pt]
      \path (a) node {};
      \path (d) node {};
      \path (e) node {};
    \end{tikzpicture}
  }
  ~
  \subfloat[]{
    \label{fig:mis-d}
    \begin{tikzpicture}
      \tikzstyle{every node}=[circle, inner sep=1.2pt, fill=darkgray]
      \path (0,0) node (a) {};
      \path (a)+(0:1) node (b) {};
      \path (b)+(-90:1) node (c) {};
      \path (c)+(-180:1) node (d) {};
      
      \draw (a)--(b)--(c)--(d)--(a);
      
      \tikzstyle{every node}=[circle, draw, red, inner sep=2pt]
      \path (a) node {};
      \path (c) node {};
    \end{tikzpicture}
  }
\vspace{-0.25cm}
  \caption{\label{fig:examples} Four examples of MISs in various graphs (resp. footprints).}
\end{figure}
For example, if the graph is a triangle (see Figure~\ref{fig:mis-a}), then only one MIS exists up to isomorphism, consisting of a single node. However, this set is no longer maximal in one of the possible connected spanning subgraphs (\eg after removing an adjacent edge to the selected node). Therefore, the triangle graph
admits no robust MIS. Some graphs do admit a robust MIS, but not all of the MISs are robust. 
Figures~\ref{fig:mis-b} and~\ref{fig:mis-c} show two MISs
in the bull graph, only one of which is robust. Finally,
 some graphs like the square graph (Figure~\ref{fig:mis-d}) are such that {\em all} MISs are robust. Although the last two examples seem to suggest that robust MISs are related to {\em maximum} MISs, being maximum is actually neither a necessary nor a sufficient condition.

In this paper, we characterize the class of graphs such that all MISs are robust, denoted \forallMIS. 
We prove that \forallMIS consists {\em exactly} of the union of complete bipartite graphs and a new class of graphs called {\em sputniks}, which contains among others all the trees (for which any property is trivially robust). While the sufficient side is easy to establish, proving that these graphs are the only ones is more difficult. Interestingly,
while the best known algorithms for deterministic distributed MIS in general graphs are superlogarithmic in the number of nodes $n$, namely they take $2^{\mathcal{O}(\sqrt{\log n})}$ rounds~\cite{PS96} (better randomized algorithms are known~\cite{LW11}), graphs in \forallMIS turn out to be specific enough to find an MIS (robust by definition) by using only information available within a sublogarithmic distance. We present an algorithm that first settles specific subsets of the networks using information available within {\em constant} distance, the residual instance being a disjoint union of trees. The residual instance can then be given to state-of-the-art algorithms like Barenboim and Elkin's for graphs of bounded arboricity~\cite{BE10}, which is known to use only information within distance $\mathcal{O}(\log n/\log \log n)$. An added benefit of this reduction is that any further progress on the MIS problem on trees will automatically transpose to robust MISs in \forallMIS. (Note that we deliberately do not use the terms ``rounds'' or ``time'', due to the non equivalence of locality and time in the context of a footprint.)

Next, we turn our attention to general graphs and ask whether a robust MIS can be found (if one exists) using only local information. We answer negatively, proving an $\Omega(n)$ lower bound on the locality of the problem. This result implies a separation between the MIS problem and the robust MIS problem in general graphs, since the former is feasible within $2^{\mathcal{O}(\sqrt{\log n})}$ hops~\cite{PS96}. It also implies that no strategy is essentially better than collecting the network at a single node and subsequently solving the problem in an offline manner. Motivated by this observation, we consider the offline problem of finding a robust MIS in a given graph if one exists (and rejecting otherwise).
The trivial strategy amounts to enumerating all MISs until a robust one is found, however there may be exponentially many MISs in general graphs (Moon and Moser~\cite{MM65}, see also~\cite{F87,GGG88} for an extension to the case of connected graphs). We present a polynomial time algorithm for computing a robust MIS in any given graph (if one exists). Our algorithm relies on a particular decomposition of the graph into a tree of biconnected components (${\cal ABC}$-tree), along which constraints are propagated about the MIS status of special nodes in between the components. The inner constraints of non trivial components are solved by reduction to the 2-SAT problem (which {\em is} tractable). As a by-product, the set of instances for which a robust MIS is found characterizes the existential analogue of \forallMIS, that is the class \existsMIS of all graphs that {\em admit} a robust MIS. (Whether a closer characterization exists is left as an open question.)

\noindent\textbf{Further discussion on the dynamic interpretation.}
As pointed out, in a dynamic network there is no way to distinguish between recurrent and non recurrent edges,
therefore the nodes cannot learn the {\em eventual} footprint \cite{DLLP16} (this observation is the very basis of the notion of
robustness). Now, what about the union of both types of edges, that is, the footprint itself? Clearly, the footprint can never be {\em decided} in a definitive sense by the nodes, since some edges may appear arbitrary late for the first time. However, it is also clear that every edge of the footprint {\em will} eventually appear; thus, over time the nodes can learn the footprint in a stabilized way, by updating their representation as new edges are detected. It is therefore possible to update some structure or property that eventually relates to the correct footprint. (Alternatively, one may assume simply that prior information about the footprint is given to the nodes, or that an oracle informs the nodes once every edge of the footprint has appeared.) Again, the reader is free to ignore the dynamic interpretation if the static one makes for a sufficient motivation.

\noindent\textbf{Outline.} Section~\ref{sec:definitions} presents the main definitions and concepts. Then, we characterize in Section~\ref{sec:forallMIS} the class \forallMIS and present a dedicated MIS algorithm that requires only information up to a sublogarithmic number of hops. Section~\ref{sec:existsMIS} establishes the non-locality of the problem in general and describes a tractable algorithm that computes a robust MIS in a given graph if one exists. Section~\ref{sec:conclusion} concludes with some remarks.

\section{Main concepts and definitions}
\label{sec:definitions}

Many of the concepts presented in the introduction, including that of temporal paths, footprint, or classes of dynamic networks are not defined here. The authors believe that the informal descriptions given in introduction are sufficient to understand the dynamic interpretation of the results. (If that is not the case, the reader is referred to~\cite{CFQS12} for thorough definitions using the time-varying graph formalism.) Our results themselves are formulated using standard concepts of graph theory, making them independent from both interpretations.

\subsection{Basic definitions}
Let $G=(V,E)$ be an undirected graph, with $V$ the set of nodes (vertices) and $E$ the set of bidirectional communication links (edges). We denote by $n=|V|$ the number of nodes in the graph, and by $D$ the {\em diameter} of the graph, that is, the length of the longest shortest path in $G$ over all possible pairs of nodes. We denote by $N(v)$ the neighborhood of a vertex $v$, which is the set of vertices $\{w : \{v,w\}\in E\}$. The degree of a vertex $v$ is $|N(v)|$. A vertex is {\em pendant} if it has degree~$1$. A {\em cut vertex} (or {\em articulation point}) is a vertex whose removal disconnects the graph. A {\em cut edge} (or {\em bridge}) is an edge whose removal disconnects the graph. We say that an edge is {\em removable} if it is not a cut edge. 
A {\em spanning connected subgraph} of a graph $G=(V_G,E_G)$ is a graph $H=(V_H, E_H)$ such that $V_H = V_G$, $E_H \subseteq E_G$, and $H$ is connected. 
In the most general variant, we define {\em robustness} as follows.

\begin{definition}[Robustness]
\label{def:robustness}
A property $P$ is said to be {\em robust} in $G$ if and only if it is satisfied in every connected spanning subgraph of $G$ (including $G$ itself).
\end{definition}

In other words, a robust property holds even after an arbitrary number of edges are removed without disconnecting the graph. Robustness is a special case of hereditary property, and more precisely a special case of decreasing monotone property (see for instance~\cite{K88}). 
In this paper, 
we focus on the {\em maximal independent set} (MIS) problem. An MIS is a set of nodes such that no two nodes in the set are neighbors and the set is maximal for the inclusion relation.
Following Definition~\ref{def:robustness}, a robust MIS in a graph $G$ (RMIS, for short) is an MIS that remains {\em maximal} and {\em independent} in every connected spanning subgraph of $G$. Observe that independence is stable under the removal of edges; therefore, it is sufficient that the MIS be maximal in all these subgraphs in order to be an RMIS. We define two classes of graphs related to the robustness of MISs.

\begin{definition}[\forallMIS]
  This class is the set of all graphs in which all MISs are robust.
\end{definition}

\vspace{-0.4cm}

\begin{definition}[\existsMIS]
  This class is the set of all graphs that admit at least one robust MIS.
\end{definition}

We define the distributed problem of computing an RMIS in a given graph as follows.

\begin{definition}[\RMIS problem]
  \label{def:rmis-problem}
  Given a graph $G$ and an algorithm ${\cal A}$ executed at every node of $G$, ${\cal A}$ solves \RMIS on $G$ iff every node eventually terminates by outputting IN or OUT, and the set of nodes outputting IN forms an RMIS on $G$. Algorithm ${\cal A}$ solves \RMIS in a class of graphs ${\cal C}$ iff for all $G \in {\cal C}$, ${\cal A}$ solves \RMIS on $G$. 
\end{definition}

Finally, let us define two classes of graphs that turn out to be closely related to RMISs, namely {\em complete bipartite graphs} and {\em sputnik graphs}. The latter is introduced here for the first time.

\begin{definition}[Complete bipartite graph]
A complete bipartite graph is a graph $G=(V_1 \cup V_2,E)$ such that $V_1 \cap V_2 = \emptyset$ and $E = V_1 \times V_2$. In words, the vertices can be partitioned into two sets $V_1$ and $V_2$ such that every vertex in $V_1$ shares an edge with every vertex in $V_2$ (completeness), and these are the only edges (bipartiteness).
\end{definition}

\vspace{-0.4cm}

\begin{definition}[Sputnik]
  A graph is a {\em sputnik} iff every vertex belonging to a cycle also has a pendant neighbor. 
\LONG{
(An example of sputnik is shown in Figure~\ref{fig:LNF}.)
}
\end{definition}

\subsection{Computational model}
Based on the chosen interpretation of our results, the base graph in the above definitions refers
either to the footprint of a dynamic network, or to the network itself. 
In the dynamic case, the actual timing of the edges is arbitrary, so the classical equivalence between time and locality in synchronous network does not hold. Nonetheless, we rely on the ${\cal LOCAL}$ model~\cite{L92,NS95} to describe the algorithms. 
To avoid confusion between locality and time in the dynamic case, we always state the complexities in terms of locality, saying that an algorithm (or problem) is $\mathcal{O}(f(G))$-local if it can be solved in $\mathcal{O}(f(G))$ rounds in the ${\cal LOCAL}$ model. (Other terminologies include saying that such problems are in $LD(f(G))$\,\cite{FKP11}.) 
For completeness, let us recall the main features of the $\mathcal{LOCAL}$ model. 
In this model, the nodes operate in synchronous discrete rounds and they wake up simultaneously. In each round, a node can exchange messages of arbitrary 
size with its neighbors and perform some local (typically unrestricted) computation.
The complexity of an algorithm over a class of graphs is the maximum number of rounds, taken over all graphs of this class, performed until all nodes have terminated. In the dynamic interpretation of our results, the algorithms are seen as being {\em restarted} every time the local knowledge of the footprint changes.

\section{Characterization of \forallMIS and locality of \RMIS}
\label{sec:forallMIS}

In this section, we show that \forallMIS, the class of graphs in which all MISs are robust, corresponds exactly to the union of complete bipartite graphs and sputnik graphs. Then we present an algorithm that solves \RMIS in \forallMIS using information available only within a sublogarithmic number of hops in $n$.

\subsection{Characterization of \forallMIS}

We first show that all MISs are robust in complete bipartite graphs and in sputnik graphs. Due to space limitation, the proofs of the two following lemmas are postponed to Appendix \ref{sec:missing}. They follow easily from the very definition of RMIS and of these classes of graphs.

\begin{lemma}
  \label{lem:BK}
  All MISs are robust in complete bipartite graphs.
\end{lemma}

\vspace{-0.4cm}

\begin{lemma}
  \label{lem:sputniks}
  All MISs are robust in sputnik graphs.
\end{lemma}


We now prove the stronger result that if a graph is such that all possible MISs are robust, then it {\em must} be either a bipartite complete graph or a sputnik. 

\begin{lemma} 
  \label{lem:necessary}
  If $G$ is not a sputnik, and yet every MIS in $G$ is robust, then $G$ is bipartite complete.
\end{lemma}

\begin{proof}
  If $G$ is not a sputnik, then some node $u$ in a cycle $C \subseteq G$ has no pendant neighbor. In general, $u$ may be an articulation point, and so the graph $G \setminus \{u\}$ may result in several components. Let $X_1, X_2, \dots$ be the resulting components with vertex $u$ back in each of them. In particular, let $X_1$ be the one that contains $C$ and observe that $X_1$ contains at least $3$ vertices (cycle). The other components, if they exist, all contain at least two vertices other than $u$ (otherwise $u$ would have a pendant neighbor). 

\noindent\textbf{Claim 1:} If all MISs in $G$ are robust, then all neighbors of $u$ in $X_1$ have the same neighborhood.

We prove this claim by contradiction.
Let two neighbors $v_1, v_2$ of $u$ be such that $N(v_1) \ne N(v_2)$. We will show that at least one MIS is not robust. Without loss of generality, assume that some vertex $x$ belongs to $N(v_1) \setminus N(v_2)$. Then we can build an MIS that contains both $v_2$ and $x$ (as a special case, $x$ may be the same vertex as $v_2$, but this is not a problem). For each of the components $X_{i\ge 2}$, choose an edge $\{u,w_i\} \in X_i$ and add another neighbor of $w_i$ to the MIS (such a neighbor exists, as we have already seen). One can see that $u$, $v_1$ and all $w_i$ can no longer enter the MIS because they all have neighbors in it. Now, choose the remaining elements of the MIS arbitrarily. We will show that the resulting MIS is not robust, by consider the removal of edges as follows. In all components $X_{i\ge 2}$, remove {\em all} edges incident to $u$ except $\{u,w_i\}$; and in $X_1$, remove all edges incident to $u$ except $\{u,v_1\}$. The resulting graph remains connected, by definition, since each of the $X_i \setminus \{u\}$ is connected. And yet, $u$ no longer has a neighbor in the MIS, which contradicts robustness.\hfill$\blacksquare$

Now, Claim 1 implies that none of $u$'s neighbors in $X_1$ has a pendant neighbor (since their neighborhoods are the same). As a result, the arguments that applied to $u$ because of its absence of pendant neighbors, apply in turn to $u$'s neighbors in $X_1$. In particular, it means that if some node $v$ is neighbor to $u$ in $X_1$, then all neighbors of $v$ (including $u$) must have the same neighborhood. Therefore, $u$ cannot be an articulation point and we are left with the single component $X_1$, in which all neighbors of $u$ have the same neighbors and these neighbors in turn have the same neighbors, which implies that the graph is complete bipartite.
\end{proof}

Based on Lemmas~\ref{lem:BK}, \ref{lem:sputniks}, and \ref{lem:necessary}, we conclude with the following theorem.

\begin{theorem}
  \label{th:forallMIS}
  All MISs are robust in a graph $G$ if and only if $G$ is complete bipartite or sputnik.
\end{theorem}

\subsection{\RMIS is locally solvable in \forallMIS}
\label{sec:forallMIS-local}

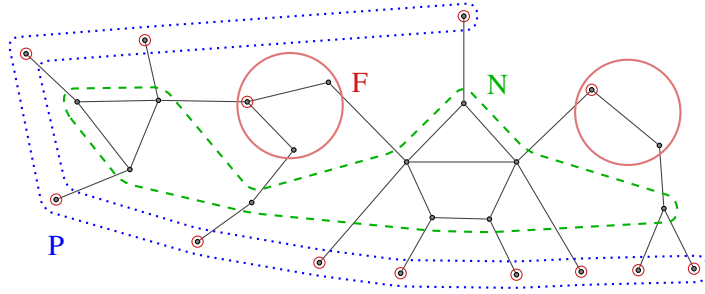
\begin{figure}
\begin{center}
\begin{tikzpicture}[xscale=1,yscale=1]
  \tikzstyle{every node}=[defnode]
  \path (1.5,10.3) node [] (v0) {};
  \path (3.08,10.48) node [] (v3) {};
  \path (1.9,8.36) node [] (v5) {};
  \path (3.78,7.8) node [] (v9) {};
  \path (6.48,7.38) node [] (v16) {};
  \path (8.02,7.36) node [] (v17) {};
  \path (7.32,10.8) node [] (v18) {};
  \path (5.4,7.52) node [] (v19) {};
  \path (8.88,7.4) node [] (v20) {};
  \path (9.64,7.42) node [] (v24) {};
  \path (10.38,7.44) node [] (v25) {};

  \path (2.18,9.66) node [] (v1) {}; 
  \path (3.26,9.68) node [] (v2) {}; 
  \path (2.88,8.76) node [] (v4) {}; 
  \path (6.56,8.86) node [] (v11) {}; 
  \path (7.32,9.64) node [] (v12) {}; 
  \path (8.02,8.86) node [] (v13) {}; 
  \path (6.9,8.12) node [] (v14) {}; 
  \path (7.66,8.1) node [] (v15) {}; 
  \path (4.5,8.32) node [] (v8) {}; 
  \path (9.98,8.24) node [] (v23) {}; 

  \path (4.44,9.66) node [] (v6) {}; 
  \path (9.02,9.82) node [] (v21) {}; 
  \path (5.06,9.02) node [] (v7) {}; 
  \path (5.52,9.92) node [] (v10) {}; 
  \path (9.92,9.08) node [] (v22) {}; 
  \tikzstyle{every path}=[];
  \draw [darkgray] (v0)--(v1);
  \draw [darkgray] (v1)--(v2);
  \draw [darkgray] (v2)--(v4);
  \draw [darkgray] (v4)--(v1);
  \draw [darkgray] (v5)--(v4);
  \draw [darkgray] (v2)--(v3);
  \draw [darkgray] (v2)--(v6);
  \draw [darkgray] (v7)--(v6);
  \draw [darkgray] (v9)--(v8);
  \draw [darkgray] (v8)--(v7);
  \draw [darkgray] (v6)--(v10);
  \draw [darkgray] (v10)--(v11);
  \draw [darkgray] (v11)--(v12);
  \draw [darkgray] (v12)--(v13);
  \draw [darkgray] (v13)--(v11);
  \draw [darkgray] (v11)--(v14);
  \draw [darkgray] (v14)--(v15);
  \draw [darkgray] (v15)--(v13);
  \draw [darkgray] (v16)--(v14);
  \draw [darkgray] (v17)--(v15);
  \draw [darkgray] (v18)--(v12);
  \draw [darkgray] (v19)--(v11);
  \draw [darkgray] (v20)--(v13);
  \draw [darkgray] (v21)--(v13);
  \draw [darkgray] (v22)--(v21);
  \draw [darkgray] (v23)--(v22);
  \draw [darkgray] (v23)--(v24);
  \draw [darkgray] (v23)--(v25);

  \tikzstyle{every node}=[mis]
  \path (v18) node [] {};
  \path (v3) node [] {};
  \path (v0) node [] {};
  \path (v5) node [] {};
  \path (v9) node [] {};
  \path (v16) node [] {};
  \path (v17) node [] {};
  \path (v19) node [] {};
  \path (v20) node [] {};
  \path (v24) node [] {};
  \path (v25) node [] {};
  \path (v6) node [] {}; 
  \path (v21) node [] {}; 

  \path (v18) coordinate[xshift=5pt,yshift=5pt] (v18ne);
  \path (v3) coordinate[yshift=5pt] (v3n);
  \path (v0) coordinate[xshift=-7pt,yshift=6pt] (v0nw);
  \path (v5) coordinate[xshift=-5pt,yshift=-4pt] (v5sw);
  \path (v9) coordinate[yshift=-5pt] (v9s);
  \path (v19) coordinate[yshift=-5pt] (v19s);
  \path (v16) coordinate[yshift=-5pt] (v16s);
  \path (v17) coordinate[yshift=-5pt] (v17s);
  \path (v20) coordinate[yshift=-5pt] (v20s);
  \path (v24) coordinate[yshift=-5pt] (v24s);
  \path (v25) coordinate[xshift=5pt,yshift=-5pt] (v25se);
  \path (v25) coordinate[xshift=5pt,yshift=5pt] (v25ne);
  \path (v24) coordinate[yshift=5pt] (v24n);
  \path (v20) coordinate[yshift=5pt] (v20n);
  \path (v17) coordinate[yshift=5pt] (v17n);
  \path (v16) coordinate[yshift=5pt] (v16n);
  \path (v19) coordinate[yshift=5pt] (v19n);
  \path (v9) coordinate[yshift=5pt] (v9n);
  \path (v5) coordinate[xshift=4pt,yshift=4pt] (v5ne);
  \path (v0) coordinate[xshift=4pt,yshift=-3pt] (v0se);
  \path (v3) coordinate[yshift=-5pt] (v3s);
  \path (v18) coordinate[xshift=5pt,yshift=-5pt] (v18se);

  \draw[thick,blue,dotted,rounded corners] (v18ne)--(v3n)--(v0nw)--(v5sw)--(v9s)--(v19s)--(v16s)--(v17s)--(v20s)--(v24s)--(v25se)--(v25ne)--(v24n)--(v20n)--(v17n)--(v16n)--(v19n)--(v9n)--(v5ne)--(v0se)--(v3s)--(v18se)--cycle;

  \path (v1) coordinate[xshift=-5pt,yshift=-3pt] (v1sw);
  \path (v4) coordinate[xshift=-3pt,yshift=-5pt] (v4sw);
  \path (v8) coordinate[yshift=-4pt] (v8s);
  \path (v14) coordinate[yshift=-5pt] (v14s);
  \path (v15) coordinate[yshift=-5pt] (v15s);
  \path (v23) coordinate[xshift=5pt,yshift=-5pt] (v23se);
  \path (v23) coordinate[xshift=5pt,yshift=4pt] (v23ne);
  \path (v13) coordinate[xshift=4pt,yshift=4pt] (v13ne);
  \path (v12) coordinate[yshift=7pt] (v12n);
  \path (v11) coordinate[xshift=-4pt,yshift=4pt] (v11nw);
  \path (v8) coordinate[yshift=4pt] (v8n);
  \path (v2) coordinate[xshift=2pt,yshift=5pt] (v2n);
  \path (v1) coordinate[xshift=-5pt,yshift=5pt] (v1nw);

  \draw[thick,green,dashed,rounded corners] (v1sw)--(v4sw)--(v8s)--(v14s)--(v15s)--(v23se)--(v23ne)--(v13ne)--(v12n)--(v11nw)--(v8n)--(v2n)--(v1nw)--cycle;

  \draw[thick,red!60] (v6)+(.57,-.05) circle (.7);

  \draw[thick,red!60] (v21)+(.48,-.3) circle (.7);

  \tikzstyle{every node}=[]
  \path (v12n) node[right=5pt, green] {N};
  \path (v5) node[below=10pt, blue] {P};
  \path (v10) node[right=5pt, red] {F};

\end{tikzpicture} 
\vspace{-0.5cm}
\end{center}
\caption{A sputnik and its sets $P$ (dotted set), $N$ (dashed set), and $F$ (plain set).}\label{fig:LNF}
\end{figure}

We now prove that computing deterministically an RMIS in class \forallMIS can be done locally, by presenting a distributed algorithm that computes a (regular) MIS using only information available within $o(\log n)$-hops a sublogarithmic number of hops in $n$. By definition of the class, this MIS is robust. 
Informally, the algorithm proceeds as follows (due to space limitations, the pseudo-code and the formal proof of the algorithm are moved to Appendix \ref{sec:missing}). Class \forallMIS consists of exactly the union of bipartite complete graphs and sputniks (Theorem~\ref{th:forallMIS}). First, the nodes decide if the graph is complete bipartite by looking within a constant number of hops (three). If so, membership to the MIS is decided according to some convention (\eg all nodes in the same part as the smallest identifier are in the MIS). Otherwise, the graph {\em must} be a sputnik and every node decides (without more information) which of the following three cases it falls into: 1) it is a pendant node (set $P$ in Figure~\ref{fig:LNF}), 2) it is not a pendant node but has at least one pendant neighbor (set $N$), or 3) none of the two cases apply (set $F$). In the first case, it enters the MIS, while in the second it decides not to. We prove that the set of nodes falling into the third case does form a disjoint union of trees, each of which can consequently be solved by state-of-the-art algorithms. In particular, Barenboim and Elkin~\cite{BE10} present a $\mathcal{O}(\log n/\log \log n)$-local algorithm that solves MIS in graphs of bounded arboricity (and a fortiori trees).
On the negative side, we show (using standard arguments) that Linial's $\Omega(\log^* n)$ lower bound for $3$-coloring~\cite{L92} in cycles extends to \RMIS in class \forallMIS, leading to the following theorem.

\begin{theorem}\label{th:localityforall}
\RMIS is $\Omega(\log^* n)\cap\mathcal{O}(\log n / \log \log n)$-local in class \forallMIS.
\end{theorem}

\section{Nonlocality of \RMIS in general graphs and global resolution}
\label{sec:existsMIS}
In this section, we prove that the problem of computing deterministically an RMIS in general graphs, if one exists, is {\em not} local. Precisely, we first observe that {\em deciding} whether an RMIS exists is not a local problem; then, we prove a $\Omega(n)$ lower bound on the distance at which it might be necessary to look to solve the problem if an RMIS exists, where $n$ is the diameter of the network. Motivated by this result, we present an offline algorithm that compute an RMIS, in polynomial time, if one exists. It can be used in a strategy where all the information about the network is collected at one node (or several, the algorithm being deterministic).

\begin{figure}
\begin{center}
\begin{tikzpicture}[scale=.7]
  \tikzstyle{every node}=[defnode]
  \path (-1,1) node [] (ma) {};
  \path (0,0) node [] (ba) {};
  \path (0,2) node [] (ha) {};
  \path (1,1) node [] (mb) {};
  \path (3,1) node [] (m0) {};
  \path (4,0) node [] (b0) {};
  \path (4,2) node [] (h0) {};
  \path (5,1) node [] (m1) {};
  \path (6,0) node [] (b3) {};
  \path (6,2) node [] (h3) {};
  \path (7,0) node [] (b3b) {};
  \path (7,2) node [] (h3b) {};
  \path (8,1) node [] (m4) {};
  \path (9,0) node [] (b4) {};
  \path (9,2) node [] (h4) {};
  \path (10,1) node [] (m5) {};
  \path (12,1) node [] (m7) {};
  \path (13,0) node [] (b7) {};
  \path (13,2) node [] (h7) {};
  \path (14,1) node [] (m8) {};

  \tikzstyle{every node}=[]
  \path (ma) node[left] {$\beta_k$};
  \path (ba) node[below] {$\gamma_k$};
  \path (ha) node[above] {$\alpha_k$};
  \path (mb) node[] {};
  \path (m0) node[above,xshift=-5pt] {$\beta_1$};
  \path (b0) node[below] {$\gamma_1$};
  \path (h0) node[above,xshift=-2pt] {$\alpha_1$};
  \path (m1) node[right,xshift=2pt] {$\beta_0$};
  \path (b3) node[below] {$\gamma_0$};
  \path (h3) node[above,xshift=2pt] {$\alpha_0$};
  \path (b3b) node[below] {$c_0$};
  \path (h3b) node[above] {$a_0$};
  \path (m4) node[left] {$b_0$};
  \path (b4) node[below] {$c_1$};
  \path (h4) node[above] {$a_1$};
  \path (m5) node[above,xshift=5pt] {$b_1$};
  \path (m7) node[] {};
  \path (b7) node[below] {$c_k$};
  \path (h7) node[above] {$a_k$};
  \path (m8) node[right] {$b_k$};
  \path (mb) node[right=9pt] {$\dots$};
  \path (m5) node[right=9pt] {$\dots$};

  \draw (ma)--(ba)--(mb)--(ha)--(ma);
  \draw (m0)--(b0)--(m1)--(h0)--(m0);
  \draw (m1)--(b3)--(b3b)--(m4)--(h3b)--(h3)--(m1);
  \draw (m4)--(b4)--(m5)--(h4)--(m4);
  \draw (m7)--(b7)--(m8)--(h7)--(m7);

  \tikzstyle{every node}=[mis]
  \path (ha) node {};
  \path (ba) node {};
  \path (h0) node {};
  \path (b0) node {};
  \path (h3) node {};
  \path (b3) node {};
  \path (m4) node {};
  \path (m5) node {};
  \path (m7) node {};
  \path (m8) node {};
\end{tikzpicture} 
\vspace{-0.65cm}
\end{center}
\caption{The graph $G_k$ ($k\in\mathbb{N}$) and one of its two possible robust MISs.}\label{fig:gk}
\end{figure}
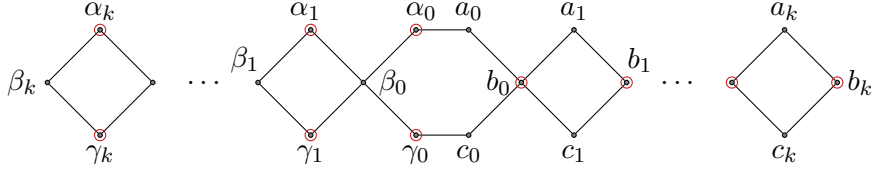

\subsection{\RMIS is non local in general graphs}\label{sub:rmisexists}

Let us first observe that the problem of deciding whether an RMIS exists is not local. Consider two graphs $G_1$ and $G_2$ which respectively consist of a $\mathcal{O}(n)$-long path and to a lollipop graph (\ie a graph joining a $\mathcal{O}(n)$-long path to a clique of size $\mathcal{O}(n)$). Then, clearly, a node at one extremity of $G_1$ and the (unique) pendant node of $G_2$ cannot distinguish their $o(n)$ neighborhood (even with identifiers, which could be exactly the same in this neighborhood) whereas $G_1$ admits an RMIS and $G_2$ does not. We go further and prove that, even if some RMISs do exist, then finding one is non local.
To prove this result, we exhibit an infinite family of graphs $(G_k)_{k\in\mathbb{N}}$, each of which has diameter linear in $k$ (and $n$). 
We first show through Lemmas~\ref{lem:gk-robust} and~\ref{lem:gk-only} that every $G_k$ admits only {\em two} RMISs $M_1$ and $M_2$ which are {\em complements} of each other; that is $M_2 = \overline{M_1} = V \setminus M_1$. 
Intuitively, these MISs are such that two nodes at distance $\mathcal{O}(n)$ must take opposite decisions, although they have the same view up to distance $\mathcal{O}(n)$. (The real proof is more complex and involves showing that identifiers do not help either.) As a result, the nodes may have to collect information up to distance $\Omega(n)$ in order to decide consistently.

Let $(G_k)_{k\in\mathbb{N}}$ be an infinite famility of graphs defined as follows. Graph 
$G_0=(V_0,E_0)$ is such that $V_0=\{a_0,b_0,c_0,\alpha_0,\beta_0,\gamma_0\}$ 
and $E_0$ induces a cycle $a_0$-$b_0$-$c_0$-$\gamma_0$-$\beta_0$-$\alpha_0$-$a_0$ as shown in Figure~\ref{fig:gk}.
Then $G_k=(V_k,E_k)$ is obtained 
from $G_{k-1}=(V_{k-1},E_{k-1})$ as follows: $V_k=V_{k-1}\cup\{a_k,b_k,c_k,\alpha_k,\beta_k,\gamma_k\}$ 
and $E_k=E_{k-1}\cup\{\{\beta_{k-1},\alpha_{k}\},\{\beta_{k-1},\gamma_{k}\},$ 
$\{\alpha_{k},\beta_{k}\},\{\gamma_{k},\beta_{k}\},\{b_{k-1},a_{k}\},
\{b_{k-1},c_{k}\},\{a_{k},b_{k}\},\{c_{k},b_{k}\}\}$.

For any $k$, define $M_1$ as the set of nodes $\{\alpha_i,\gamma_i,b_i|i\le k\}$ and $M_2=\{a_i,c_i,\beta_i|i \le
k\}$. Observe that $M_2 = V_k\setminus M_1$ (written $\overline{M_1}$). Set $M_1$ is illustrated in Figure~\ref{fig:gk}.

\begin{lemma}\label{lem:gk-robust}
For any $k\geq 0$, 
$M_1$ and 
$M_2$ 
are RMISs in $G_k$.
\end{lemma}

\begin{proof}
  We prove this for $M_1$. The same holds symmetrically for $M_2$.
  First, observe that $M_1$ is a valid MIS: no two of its nodes are neighbors by construction (independence) and all nodes in $V_k\setminus M_1$ have neighbors in $M_1$ (maximality).
Now, to obtain a connected spanning subgraph of $G_k$, one
can remove from $E_k$ at most one edge from each simple cycle of $G_k$. Since any node of
$V_k\setminus M_1$ has a number of neighbors in $M_1$ strictly greater than the number of simple cycles it belongs to, $M_1$ is robust.
\end{proof}

\begin{lemma}\label{lem:gk-only}
For any $k\geq 0$, $M_1$ and $M_2$ are the only two RMISs in $G_k$, implying that nodes $b_k$ and $\beta_k$ must take opposite decisions in all RMISs.
\end{lemma}

\begin{proof}
We say that an edge $e \in E$ is \emph{critical} with respect to some MIS $M$ in $G_k$ 
if $e$ is removable (\ie not a cut edge) and $M$ is no longer maximal in $(V_k,E_k\setminus\{e\})$. The existence of a critical edge implies that the considered MIS is not robust.
Let us now consider an RMIS $M$ in $G_k$. We prove several claims 
on $M$.

\noindent\textbf{Claim 1:} If $\alpha_0\in M$, then $\{\gamma_0,b_0\} \subseteq M$ and $\{\beta_0,c_0,a_0\} \subseteq \overline{M}$.

If $\alpha_0\in M$, then $\{\beta_0,a_0\} \subseteq \overline{M}$ (independence). It also holds that $b_0\in M$, otherwise the edge $\{\alpha_0,a_0\}$ is critical (robustness). It follows that $c_0\notin M$ (independence) and $\gamma_0\in M$ (maximality).

\noindent\textbf{Claim 2:} If $\alpha_0\notin M$, then $\{\gamma_0,b_0\} \subseteq \overline{M}$ and $\{\beta_0,c_0,a_0\} \subseteq M$. \hfill {\it (symmetric to Claim~1)}

\noindent\textbf{Claim 3:} If $\alpha_0\in M$, then $\{\alpha_i,\gamma_i\}\subseteq M$ and $\beta_i\notin M$ for all $i\le k$.

By contradiction, if $\beta_i\in M$ for some $i$, then $\{\alpha_i,\gamma_i\} \subseteq \overline{M}$ (independence). Let $i$ be smallest possible, then edges $\{\alpha_i,\beta_i\}$ and 
$\{\beta_i,\gamma_i\}$ are critical \wrt $M$ (recall that, if 
$i>0$, $\beta_{i-1}\notin M$ by construction and $\beta_0\notin M$ 
by Claim 1), which contradicts robustness. Therefore, $\beta_i\notin M$.
The maximality of $M$ allows us to conclude.

\noindent\textbf{Claim 4:} If $\alpha_0\notin M$, then $\{\alpha_i,\gamma_i\}\subseteq \overline{M}$ and $\beta_i\in M$ for all $i\le k$. \hfill {\it (symmetric to Claim 3)}

\noindent\textbf{Claim 5:} If $\alpha_0\in M$, then $\{a_i,c_i\}\subseteq \overline{M}$ and $b_i\in M$ for all $i\le k$.

By contradiction, if $b_i\notin M$ for some $i$, then $\{a_i,c_i\}\subseteq \overline{M}$ (independence). Let $i$ be smallest possible (recall that, if 
$i>0$, $b_{i-1}\in M$ by construction and $b_0\in M$ 
by Claim 1). Let $w$ be $b_0$ if $i=0$ and be $b_{i-1}$ otherwise.
The edges $\{w,a_i\}$ and $\{w,\gamma_i\}$
are then critical, which contradicts robustness. Therefore, $b_i\in M$.
The independence of $M$ allows us to conclude.

\noindent\textbf{Claim 6:} If $\alpha_0\notin M$, then $\{a_i,c_i\}\subseteq M$ and $b_i\notin M$ for all $i\le k$. \hfill {\it (symmetric to Claim 5)}

Claims 1 to 6 imply that $M=M_1$ if $\alpha_0\in M$ and $M=M_2$ otherwise.
\end{proof}

Finally, we relate these results to the locality of the \RMIS problem.

\begin{theorem}\label{th:localityexists}
\RMIS requires the nodes to use information up to distance $\Omega(n)$ in $G_k$.
\end{theorem}

\begin{proof}
The proof would be straightforward in an anonymous network, due to the fact that $b_k$ and $\beta_k$ have indistinguishable structural neighborhoods (\aka views~\cite{YK96}) up to distance $\mathcal{O}(n)$, and yet, they must take different decisions (Lemma~\ref{lem:gk-only}). Unique identifiers make the argument more complicated, since $b_k$ and $\beta_k$ do have unique {\em labeled} views (\ie views taking into account identifiers) even at distance $0$. 

Let us call $b_k$ and $\beta_k$ the {\em extremities} of the network. Observe that the distance between both extremities is larger than $4k$.
Let ${\cal L}_1$, ${\cal L}_2$, and ${\cal L}_3$ be three possible labeling functions that assign unique identifiers to the neighborhood of an extremity up to distance $k$ (say) and such that the three labelings have no identifier in common. Let $G_k^1$ be the labeled graph whose structure is isomorphic to $G_k$, in which the neighborhood of $b_k$ is labeled according to ${\cal L}_1$ and the neighborhood of $\beta_k$ is labeled according to ${\cal L}_2$; the rest of the nodes are labeled arbitrarily. Let $G_k^2$ be defined similarly, but using ${\cal L}_3$ instead of ${\cal L}_2$ in the neighborhood of $\beta_k$. Finally, let $G_k^3$ be defined similarly, but using ${\cal L}_2$ in the neighborhood of $b_k$ and ${\cal L}_3$ in the neighborhood of $\beta_k$. Now, if $b_k$ and $\beta_k$ use only information up to distance $k$, then they must take identical decisions in at least one of the three labeled graphs, contradicting Lemma~\ref{lem:gk-only}.
\end{proof}

\subsection{A global algorithm to compute a robust MIS (if one exists)}

We now describe an algorithm that tests constructively whether an RMIS exists in a graph $G$. 
Our algorithm relies on the construction of an auxiliary tree called {\em \ABCT}, which represents a particular decomposition of
the graph based on biconnected components (it is neither a block-cut tree, nor a bridge tree, but a mix of these two types of decomposition). Roughly speaking, our algorithm works by propagating constraints about the MIS along 
the \ABCT . 
Each non-trivial component is solved on the way up by means of a reduction of its constraints to 2-SAT (which {\em is} polynomial-time solvable). 

In the following, we describe how the \ABCT is built over $G$. It is followed by 
an informal presentation of the algorithm---due to the lack of space, the formal algorithm has been moved to
Appendix~\ref{sec:pseudo} and its proof is presented in Appendix~\ref{sec:bigproof}. 

\noindent\textbf{Decomposition of $G$.} 
In the context of this section, we call {\em biconnected component} (or simply {\em component}) in $G$ a maximal
subgraph $H\subseteq G$ such that the removal of any node in $H$ does not disconnect $H$ (\ie $H$ is $2$-{\it vertex}-connected). 
By abuse of notation, we write $u \in H$ if $u$ is a vertex of the subgraph $H$. 
We consider here a mix of the so-called block-cut tree and bridge tree and refer to it as the {\it \ABCT}. 
Let \BS be the set of biconnected components of $G$ (see Figure~\ref{fig:ex_compo} for an illustration). 
\begin{figure}
\begin{center}
  \includegraphics[width=.55\textwidth] {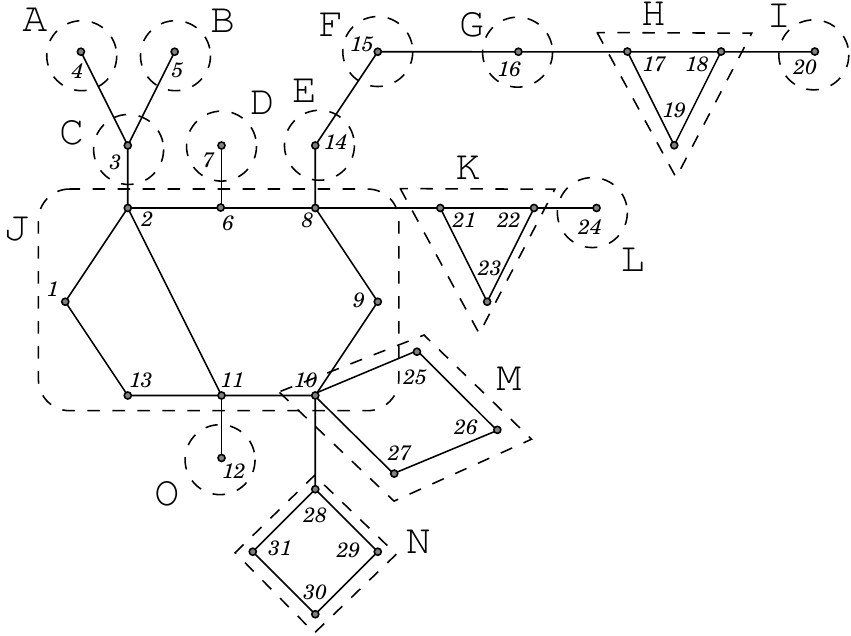}
  \hspace{-35pt}
  \includegraphics[width=.47\textwidth]{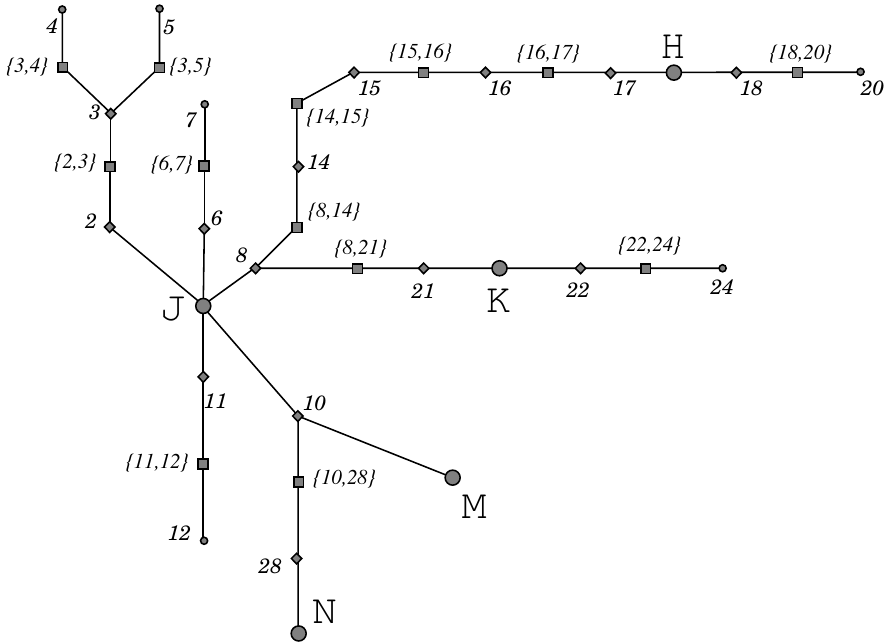}
\end{center}
\vspace{-0.65cm}
\caption{\label{fig:ex_compo}
  Decomposition of a graph into biconnected components (left) and the corresponding \ABCT (right).
	Nodes with shapes~$\circ$ and~{\scriptsize $\Diamond$} correspond respectively to pendant vertices ($\in \PV$) and articulation points ($\in \A$). 
	Nodes with shape~{\tiny $\Box$} correspond to bridge edges ($\in \B$), and  nodes with shape~{\tiny $\bigcirc$} show components that contain at least $3$ vertices ($\in \C$).
}
\end{figure}

Two adjacent components either share a common articulation point or they are linked by a bridge. 
For instance, node $10$ in Figure~\ref{fig:ex_compo} is common to components $J$ and $M$, and edges $\{6,7\}$ and $\{8,21\}$ are bridges between $D$ and $J$, and $J$ and $K$ respectively.  
Let \A be the set of all articulation points (whether or not they are shared) and \B the set of bridges. Let \PV be the set of all pendant vertices, which form singleton components---in
Figure~\ref{fig:ex_compo}, $\PV = \{4,5,7,12,20,24\}$. Finally, let \C be the set of components that contain three or more vertices---$\C = \{H,J,K,M,N\}$.  
Remark that if \C is empty, then $G$ is acyclic.

The ${\cal ABC}$ decomposition (or \ABCT) of $G$, is the graph $\CT=(V', E')$ such that $V' = \A \cup \B \cup \C \cup \PV$ and $E'$ is
defined by the two following rules:
$(i)$ $\forall a \in \A$, $\forall c \in \C$, $\{a,c\} \in E'$ if and only if $a \in c$; and
$(ii)$ $\forall b=\{u,v\} \in \B$, $\{b,u\} \in E'$ and $\{b,v\} \in E'$.
Figure~\ref{fig:ex_compo} (right side) shows the \ABCT corresponding to the graph of the left side.


\noindent\textbf{Algorithm.}
The algorithm works on \CT, the \ABCT made over the four sets \A, \B, \C, and \PV with respect to $G$.  
If the set \C is empty, it means that $G$ is acyclic.  In that case, $G$ trivially admits an RMIS, 
which is returned by the algorithm. 
Otherwise ($\C  \neq \emptyset$), a component vertex $C \in \C$ is arbitrarily selected to be the root of \CT,
denoted by \R.  
Then, the classical concepts in oriented trees, such as {\em children}, {\em parent},
{\em descendant}, {\em subtree}, or {\em leaf} apply to the vertices of \CT.  
For ease of reading, we abuse the term ``{\em admit}'' by saying that a subtree $\hat{T}$ of \CT ``{\em admits an
RMIS}'' if the subgraph of $G$ corresponding to $\hat{T}$ admits an RMIS.

At the higher level, the algorithm proceeds within two phases.  Based on \CT, the first one is called the {\em
labeling} phase. 
Initiated from the leaves of the tree, it evaluates whether the subgraph of $G$ corresponding to the current 
subtree admits a robust MIS or not and labels each subtree according to that. 
If it does admit an RMIS, it may impose some constraints about the membership of 
the higher nodes.  For instance, a robust MIS of the subtree may exist only if the articulation point leading up to the parent belongs to it. 
Then, the goal of the labeling phase consists of propagating (and memorizing within labels) such constraints up through subtle interactions among the
various types of vertices (namely, pendant nodes, articulation point, bridge, or component) leading up to the root.  
Besides, the inner topological configuration of a single component may also impose non-trivial 
constraints for the existence of a robust MIS. Intuitively, it must have properties that relate to (but are slightly more complex than) 
bipartiteness. 
The second (short) phase of the algorithm is called the {\em deciding} phase.
It simply consists of deciding whether the graph admits a robust MIS or not considering 
the label of the root of the \ABCT. 

\noindent\textbf{Labeling Phase.}
As already mentioned, a given subtree $\hat{T}$ of \CT may or may not admit an RMIS. 
Intuitively, the global decision depends on topological contraints, established over $\hat{T}$.  Obviously, those constraints
influence the possible topological organization of a global RMIS toward the parent of $\hat{T}$. So, this mecanism
involves in a crucial way the unique $x \in \hat{T}$ such that $x \in \A \cup \PV$ through which $\hat{T}$
is connected to the remainder of \CT.   In the following $x$ is called 
the {\em attachment point} of $\hat{T}$. In other words and more conveniently, the attachment point of the root $r$ of $\hat{T}$ is
either $r$ itself if $r \in \A \cup \PV$ or the parent of $r$ if $r \in \B \cup \C$.  
For instance, in 
Figure~\ref{fig:ex_compo}, assuming that \R is the component vertex {\tt M}, then, $4$ is the
attachment point of itself ($4 \in \PV$), $8$ is the one of both $\{8,14\}$ and $\{8,21\}$, while $10$ is the
attachment point of $\{10,28\}$ ($\in \B$), $M$ ($\in \C$), and itself ($\in \A)$. 

Constraint transmission takes place from the leaves to the root by tagging 
each vertices with the following labels: \typePI for \typePIdesc (meaning that $\hat{T}$ admits an RMIS that includes $x$, the attachment point of $\hat{T}$); \typePO  for \typePOdesc (meaning that $\hat{T}$ admits an RMIS that does not include $x$); \typePE for \typePEdesc (meaning that $\hat{T}$ is not tagged \typePO, and it admits an RMIS that does not include $x$ assuming that another node $x'$, external to $\hat{T}$ 
belongs to the RMIS); and \typeN for \typeNdesc (meaning that none of the three other tags is applicable to $\hat{T}$). An extra label, \typeE is used at the root (see below). Note that the algorithm associates to each label a set of vertices that is used to store a robust MIS of $\hat{T}$ satisfying the constraint of the label.  Also, remark that \typePO and \typePE are mutually exclusive.  
Furthermore, a vertex can be tagged with more than a single tag, namely either \typePI and \typePO
(together), or \typePI and \typePE (together).  

The analysis consists of recursively tagging each vertex $v \in \CT$ from the leaves to the root.  Let us first
consider $v$ as a leaf.  There are two cases: either $v \in \PV$ or $v \in \C$.  In the former case, 
$v$ is tagged with both \typePI and \typePE.  (Indeed, $v$ being a pendant node, $v$ can or cannot be in the RMIS
depending or its unique neighbor.) For instance, in 
Figure~\ref{fig:ex_compo}, assuming that $\R = M$, the vertices
$4$, $5$, and $20$ are tagged \typePI and \typePE.  
If $v \in \C$, then the algorithm checks whether $v$ must be tagged \typeN, \typePI, \typePO, or \typePE.  For instance, in
the same example, $N$ is tagged \typePI and \typePO.  Indeed, $N$ is a square (sub)graph (refer to Figure~\ref{fig:ex_compo}) and
its attachment point is $28$.  Clearly, it admits two possible RMISs: either $\{31,29\}$ or $\{28,30\}$. In former
case, $28$ does not belong to the RMIS (implying the label of type \typePO); in the latter, $28$ belongs to the RMIS (implying the label of type \typePI). 
The actual finding is solved through reduction to the 2-SAT problem described below.

From now on, consider that $v$ is an internal vertex (\ie $v \notin \PV$).  
Provided that none of its descendants is tagged \typeN, an internal vertex $v \in \CT$ is analyzed according of its
type (whenever a vertex has two tags, each corresponding rule is applied) as follows.

Consider first the case where $v \in \B$. If its (unique) descendant $u$ is tagged \typePI,
 then $v$ is tagged 
\typePO; if $u$ is tagged \typePE, then $v$ is tagged \typePI; if $u$ is tagged \typePO, then 
$v$ is tagged \typePI, and if possible ($v$ is not already tagged \typePO), also \typePE.
In the example (Figure~\ref{fig:ex_compo}), the vertices $\{3,4\}$, $\{3,5\}$, and  $\{18,20\}$ are
all tagged \typePI and \typePO.

When $v \in \A$, let $D_v$ be the set of descendant vertices of $v$. If every vertex $u \in D_v$ is
tagged \typePI, then $v$ is also tagged \typePI; if every vertex $u \in D_v$ is
tagged \typePE, then $v$ is also tagged \typePE; if every vertex $u \in D_v$ is
tagged either \typePO or \typePE and there exists 
$u' \in D_v$ tagged \typePO, then $v$ is tagged \typePO. 
For instance, vertices $3$, $11$, and $18$ are all tagged \typePI and \typePO.

If $v \in \C\setminus\{\R\}$, then, as for the leaves, $v$ is analysed using the 2-SAT reduction described below. 
However, by contrast with a leaf vertex, the analysis introduces extra
clauses to the 2-SAT expression, according to the labels of its descendants.  

If $v = \R$ ($\in \C$), then the algorithm operates the last 2-SAT reduction that checks 
the existence of an RMIS, again accoring to the labels of its descendants. \R is then tagged
either with \typeN ($G$ admits no RMIS) or with \typeE ($G$ admits at least one RMIS).
%

\noindent\textbf{Testing a component vertex.}
The finding process is mainly based on the resolution of a 2-SAT expression.
Let us first consider an internal $C$ component---leaf and root components are special cases that are addressed later.
The procedure first consider $C'$ that is equal to $C$ to which edges linking vertices both tagged \typePO are removed.  Doing so,
$C'$ may be split into several components. If $C'$ is not bipartite, then the component is tagged with \typeN.  Otherwise, for each maximal connected component $C_i$, one part of the bipartition is arbitrarily
chosen in which each vertex $v$ ($\in C'$) is labelled with a label $\ell(v)$ equal to $x_i$.  
The vertices of the other part are labelled $\neg x_i$.  
All those labels form a 2-SAT expression to which the tags coming up from the subtrees are included.  
For instance, a node $v \in \A$ that is tagged
\typePI forces the label $\ell(v)$ of the corresponding vertex $v \in C_i$ to {\tt true}.  Also, the edge $\{u,v\}$ that was 
removed from $C$ to $C'$ also forces the labels corresponding to $u$ and $v$ to be mutually exclusive
($\neg\ell(u)\vee\neg\ell(v)$), meaning that at most one of the two can be included into the RMIS, but not both.   

Since a vertex of \CT can be tagged with one or two tags, the satisfiability of the 2-SAT expression is evaluated
assuming first that the attachment point $x$ of $C$ belongs to the RMIS (Tag~\typePI, $\ell(x) = \mathtt{true}$).  
Next, it is evaluated assuming that the attachment point $x$ of $C$ does not belong to the RMIS 
(Tag~\typePO, $\ell(x) = \mathtt{false}$).
If $C$ could not be tagged \typePO (\ie the expression could not be satisfied assuming $\ell(x) = \mathtt{false}$), 
then it can still be tagged \typePE.  This is done 
by temporarily adding an aerial $\mathsf{g}$ at the attachment point $x$ of $C$ (\ie a virtual extra vertex $y$ with 
the corresponding edge $\{x,y\}$) and repeating the whole above process with $C \cup \mathsf{g}$.

Note this process is also performed at the leaves and at the root ($C = \R$). However, in both cases, the process is 
simpler.  Indeed, since the leaves have no descendant, it is sufficient to check whether $C$ is bipartite or not.  For
the root (that has no attachment point), it is sufficient to check whether the 2-SAT expression is satisfied or not.

\noindent\textbf{Deciding phase.} 
This phase simply consists in testing the label of \R (attributed
in the labeling phase). If \R is labeled with \typeN, the algorithm rejects. Otherwise (\R
is labeled with \typeE), the algorithm returns the set associated to the label of \R that 
is a robust MIS of $G$ thanks to the work done in the labeling phase.

\section{Conclusion}
\label{sec:conclusion}

This paper is dedicated to showing the actual impact of robustness in highly dynamic distributed systems. 
A property is robust if and only if it is satisfied in every connected spanning subgraphs of a given graph. 
Focusing on the minimal independent set (MIS) problem, we proved the existence of a significant complexity gap between graphs where \emph{all} MIS are robust (building a robust MIS is then a \emph{local} problem) and graphs where \emph{some} MIS are robust (building a robust MIS is then a \emph{global} problem).

We are convinced that robustness is a key property of highly dynamic systems 
to achieving stable structures in such unstable environments.  
The complete characterization of the class \existsMIS is left open, as well as the study of similar symmetry breaking tasks.  

\newpage

\pagenumbering{roman}
\setcounter{page}{1}


\newpage

\appendix

\section{Missing proofs}
\label{sec:missing}


\paragraph{\bf{Proof of Lemma~\ref{lem:BK}}}
({\it All MISs are robust in complete bipartite graphs.})

\begin{missingproof}
  There are two ways of chosing an MIS in a complete bipartite graph $G=(V_1\cup V_2, E)$, namely $V_1$ or $V_2$. Without loss of generality, assume $V_1$ is chosen. Then, in any connected spanning subgraph of $G$, every node in $V_2$ has at least one neighbor in $V_1$ (the graph is bipartiteness). So the MIS remains maximal. (Independence is not affected, as discussed in Section~\ref{sec:definitions}.)
\end{missingproof}

\paragraph{\bf{Proof of Lemma~\ref{lem:sputniks}}}
({\it All MISs are robust in sputnik graphs.})

\begin{missingproof}
  By definition, any removable edge in a sputnik graph belongs to a cycle, thus both of its endpoints have a pendant neighbor. On the other hand, it holds that a pendant node either is in the MIS, or its neighbor must be (thanks to maximality). As a result, after an edge is removed, both of its endpoints remain covered by the MIS, \ie either they are in the MIS or their pendant neighbor is, which preserves maximality.
\end{missingproof}

\paragraph{\bf{Proof of Theorem~\ref{th:localityforall}}}
({\it \RMIS is $\Omega(\log^* n)\cap\mathcal{O}(\log n / \log \log n)$-local in class \forallMIS.})

\begin{missingproof}\noindent (Lower bound):
If one can solve \RMIS in \forallMIS looking at distance $o(\log^* n)$, then as a special case, one solves regular MIS in paths (Indeed, RMISs are MISs, and paths are trees, which are sputniks, which all belong to \forallMIS). 
Then, one can convert the resulting MIS into a $3$-coloring as follows: each node in the MIS takes color $1$, then at most two nodes lie between these nodes. The one with smallest identifier takes color $2$ and the other takes color $3$. Finally, it is well known that Linial's $\Omega(\log^* n)$ lower bound for $3$-coloring~\cite{L92} in cycles extends to paths (within an additive constant), which gives the desired contradiction.
\end{missingproof}

\begin{algorithm}[h]
\caption{Computing an RMIS in \forallMIS for node $v$.}\label{algo:rmis}
\begin{tabbing}
XX: \= XX \= \kill
01: \> Collect the information available within distance $3$ from $v$ in $B_v$\\
02: \> \textbf{If} $B_v$ is a complete bipartite graph with no outgoing edges towards 
nodes not in $B_v$ \textbf{Then}\\
03: \> \> Set $\ell$ as the node with the lowest identifier of $B_v$\\
04: \> \> Terminate outputting $IN$ if $v$ is in the same part as $\ell$, $OUT$ otherwise\\
05: \> \textbf{If} $v$ is a pendant node \textbf{Then}\\
06: \> \> Terminate outputting $IN$\\
07: \> \textbf{If} $v$ has a pendant neighbor \textbf{Then}\\
08: \> \> Terminate outputting $OUT$\\
09: \> Set $E_v$ as the set of outgoing edges of $v$ towards nodes that have a pendant neighbor \\
10: \> Execute the MIS algorithm from~\cite{BE10} ignoring edges of $E_v$
\end{tabbing}
\end{algorithm}

\begin{missingproof}\noindent (Upper bound):
Let $G=(V,E)$ be a graph in \forallMIS. We will prove that Algorithm \ref{algo:rmis} computes a regular MIS (which by definition is robust) in $G$.
All nodes first gather information within distance three (Line 1) and decides if the graph is complete bipartite (Line 2). Three hops are sufficient because all the nodes in a complete bipartite graph are at most at distance 2. Both parts of the test (bipartiteness and completeness) are trivial. If the test is positive, then all nodes which are in the same part as the smallest identifier output IN, the others OUT (lines 3 and 4). Since $G$ is bipartite, the set of nodes outputting $IN$ is independent, and since $G$ is complete, it is maximal.
Now, if the graph is not complete bipartite, then it is a sputnik (Theorem \ref{th:forallMIS}). Let us partition the set of nodes of $G$ into three parts (refer to Figure \ref{fig:LNF} for an illustration): 
the set $P$ of pendant nodes in $G$ (i.e. nodes of degree one), the set $N$ of nodes with at least one neighbor in $P$, and the 
set $F=V\setminus (P\cup N)$ of the nodes which are neither in $P$ or in $N$. Observe that it is easy for a node to determine which of the sets it belongs to, based on the ($3$-hop) information it already has. Furthermore, two neighbors cannot be in $P$, since this would imply a single-edge graph that would then be classified as complete bipartite in the previous step. Then, nodes in $P$ terminate outputting $IN$ (Line 6) and nodes in $N$ terminate outputting $OUT$ (Line 8). The crucial step is that the subgraph of $G$ induced by the nodes of $F$ (denoted 
$G_F$ in the following) is a forest due to the definition of a sputnik. Indeed, by definition, any node involved in a cycle has at least a pendant neighbor and thus belong to $N$. Furthermore, these nodes impose no constraint onto the remaining nodes in $F$ since they are {\em not} in the MIS, and yet, they do not need additional neighbors to be. As a result, the edges between $F$ and $N$ can be ignored (Line 9) and a generic MIS algorithm be executed on the induced forest $G_F$. One such algorithm~\cite{BE10}, dedicated to graphs of bounded arboricity (which the case of $G_F$) requires looking only within $\mathcal{O}(\log n/\log \log n)$ hops. Note that variable $n$ in this formula corresponds to the number of nodes in $G_F$, which is dominated by $|V|$.

Finally, we will prove that the produced MIS is
valid in $G$. Let us call $M$ this MIS, and $M_F \subseteq M$ the MIS produced by algorithm~\cite{BE10} on $G_F$. Then clearly $M=P\cup M_F$ (all nodes in $N$ output
$OUT$). $M$ is independent since $(i)$ $M_F$ is independent in
$G_F$ (and thus in $G$); $(ii)$ no node in
$F$ is neighbor to a node in $P$ (by construction); and $(iii)$ no
two nodes of $P$ are neighbors (as already discussed). As to the maximality, if there exists an independent set $M'=M\cup\{u\}$ for some $u$ in $V\setminus M$, then $u$ must belong to either $P, N,$ or $F$. Being in $P$ is not possible since all nodes in $P$ are already in $M$. Being in $N$ contradicts independence of $M'$ since
any node in $N$ is neighbor to at least one node in $P$ (that belongs to $M$).
Finally, being in $F$ contradicts the fact that $M_F$ is maximal, which concludes the proof.
\end{missingproof}

\newpage

\section{Pseudo-code of Algorithm of Section~\ref{sec:existsMIS}}
\label{sec:pseudo}

\begin{algorithm}[h]
\caption{\textbf{FindRMIS}}\label{algo:findrmis}
\textbf{Input:} A graph $G=(V,E)$\\
\textbf{Output:} A robust MIS of $G$ if $G$ admits one, $\emptyset$ otherwise
\begin{tabbing}
XX: \= XX \= XX \= XX \= XX \= XX \= XX \= XX \=\kill
01: \> Build $\CT=(\A\cup\B\cup\C\cup\PV,E')$ be the \ABCT of $G$\\
02: \> \textbf{If} $\C=\emptyset$ \textbf{then}\\
03: \> \> Build a 2-coloring of $G$\\
04: \> \> Return one non empty maximal set of nodes of $V$ sharing the same color\\
05: \> Let $r\in\C$ (arbitrarily choosen)\\
06: \> Root \CT towards $r$\\
07: \> Let $C(x)$ be the set of children in \CT of each $x\in\A\cup\B\cup\C\cup\PV$\\
08: \> Let $P(x)$ be the parent in \CT of each $x\in\A\cup\B\cup\C\cup\PV\setminus\{r\}$\\
09: \> Associate an empty set of labels $L(x)$ to each $x\in\A\cup\B\cup\C\cup\PV$\\
\> \> \> (a label is a couple $(t,S)$ with type $t\in \{\typePI,\typePO,\typePE,\typeN,\typeE\}$ and $S\subseteq V$)\\
10: \> \textbf{Foreach} $c\in C(r)$ \textbf{do}\\
11: \> \> \textbf{LabelSubTree}$(\CT,c)$\\
12: \> \textbf{If} $L(c)=\{(\typeN,\emptyset)\}$ for a $c\in C(r)$ \textbf{then}\\
13: \> \> $L(r):=\{(\typeN,\emptyset)\}$\\
14: \> \textbf{Else}\\
15: \> \> $R:=$\textbf{TestRMIS}$(\CT,r,\emptyset,\emptyset)$\\
16: \> \> \textbf{If} $R=\bot$ \textbf{then}\\
17: \> \> \> $L(r):=\{(\typeN,\emptyset)\}$\\
18: \> \> \textbf{Else}\\
19: \> \> \> $L(r):=\{(\typeE, R)\}$\\
20: \> \textbf{If} $L(r)=\{(\typeN,\emptyset)\}$ \textbf{then}\\
21: \> \> Return $\emptyset$\\
22: \> \textbf{Else}\\
23: \> \> Return the set of the label of type \typeE of $L(r)$
\end{tabbing}
\end{algorithm}

\begin{algorithm}[h]
\caption{Function \textbf{LabelSubTree}$(\CT,x)$}\label{algo:labelsubtree}
\textbf{Parameters:} An \ABCT \CT and a node $x$ of \CT\\
\textbf{Return:} None
\begin{tabbing}
XX: \= XX \= XX \= XX \= XX \= XX \= XX \= XX \=\kill
01: \> \textbf{Foreach} $c\in C(x)$ \textbf{do}\\
02: \> \> \textbf{LabelSubTree}$(\CT,c)$\\
03: \> \textbf{If} $L(c)=\{(\typeN,\emptyset)\}$ for a $c\in C(x)$ \textbf{then}\\
04: \> \> $L(x):=\{(\typeN,\emptyset)\}$\\
05: \> \textbf{Else}\\
06: \> \> \textbf{If} $x\in \A$ \textbf{then}\\
07: \> \> \> \textbf{LabelNodeA}$(\CT,x)$\\
08: \> \> \textbf{If} $x\in \B$ \textbf{then}\\
09: \> \> \> \textbf{LabelNodeB}$(\CT,x)$\\
10: \> \> \textbf{If} $x\in \C$ \textbf{then}\\
11: \> \> \> \textbf{LabelNodeC}$(\CT,x)$\\
12: \> \> \textbf{If} $x\in \PV$ \textbf{then}\\
13: \> \> \> $L(x):=\{(\typePI,\{x\}),(\typePE,\emptyset)\}$
\end{tabbing}
\end{algorithm}

\begin{algorithm}[h]
\caption{Function \textbf{LabelNodeA}$(\CT,x)$}\label{algo:labelnodea}
\textbf{Parameters:} An \ABCT \CT and a node $x\in\A$ of \CT\\
\textbf{Return:} None
\begin{tabbing}
XX: \= XX \= XX \= XX \= XX \= XX \= XX \= XX \=\kill
01: \> \textbf{If} $L(c)$ contains a label $(\typePI,R_c)$ for each $c\in C(x)$ \textbf{then}\\
02: \> \> $L(x):=L(x)\cup\{(\typePI,\bigcup_{c\in C(x)}R_c)\}$\\
03: \> \textbf{If} $L(c)$ contains a label $(\typePE,R_c)$ for each $c\in C(x)$ \textbf{then}\\
04: \> \> $L(x):=L(x)\cup\{(\typePE,\bigcup_{c\in C(x)}R_c)\}$\\
05: \> \textbf{If} $L(c)$ contains a label $(\typePO,R_c)$ or $(\typePE,R_c)$ for each $c\in C(x)$\\
\> \> \> \textbf{and} $L(c)$ contains a label $(\typePO,R_c)$ for a $c\in C(x)$ \textbf{then}\\
06: \> \> $L(x):=L(x)\cup\{(\typePO,\bigcup_{c\in C(x)}R_c)\}$
\end{tabbing}
\end{algorithm}

\begin{algorithm}[h]
\caption{Function \textbf{LabelNodeB}$(\CT,x)$}\label{algo:labelnodeb}
\textbf{Parameters:} An \ABCT \CT and a node $x=\{P(x),c\}\in\B$ of \CT\\
\textbf{Return:} None
\begin{tabbing}
XX: \= XX \= XX \= XX \= XX \= XX \= XX \= XX \=\kill
01: \> \textbf{If} $L(c)$ contains a label $(\typePI,R_c)$ \textbf{then}\\ 
02: \> \> $L(x):=L(x)\cup\{(\typePO,R_c)\}$\\
03: \> \textbf{If} $L(c)$ contains a label $(\typePO,R_c)$ \textbf{then}\\
04: \> \> $L(x):=L(x)\cup\{(\typePI,\{P(x)\}\cup R_c),(\typePE,R_c)\}$\\
05: \> \textbf{If} $L(c)$ contains a label $(\typePE,R_c)$ \textbf{then}\\
06: \> \> $L(x):=L(x)\cup\{(\typePI,\{P(x)\}\cup R_c)\}$
\end{tabbing}
\end{algorithm}

\begin{algorithm}[h]
\caption{Function \textbf{LabelNodeC}$(\CT,x)$}\label{algo:labelnodec}
\textbf{Parameters:} An \ABCT \CT and a node $x\in \C$ of \CT\\
\textbf{Return:} None
\begin{tabbing}
XX: \= XX \= XX \= XX \= XX \= XX \= XX \= XX \=\kill
01: \> $R:=$\textbf{TestRMIS}$(\CT,x,\{P(x)\},\emptyset)$\\
02: \> \textbf{If} $R\neq\bot$ \textbf{then}\\
03: \> \> $L(x):=L(x)\cup\{(\typePI,R)\}$\\
04: \> $R:=$\textbf{TestRMIS}$(\CT,x,\emptyset,\{P(x)\})$\\
05: \> \textbf{If} $R\neq\bot$ \textbf{then}\\
06: \> \> $L(x):=L(x)\cup\{(\typePO, R)\}$\\
07: \> \textbf{Else}\\
08: \> \> $L(P(x)):=\{(\typePO,\emptyset)\}$\\
09: \> \> $R:=$\textbf{TestRMIS}$(\CT,x,\emptyset,\emptyset)$\\
10: \> \> \textbf{If} $R\neq\bot$ \textbf{then}\\
11: \> \> \> $L(x):=L(x)\cup\{(\typePE,R)\}$\\
12: \> \> $L(P(x)):=\emptyset$\\
13: \> \textbf{If} $L(x)=\emptyset$ \textbf{then}\\
14: \> \> $L(x):=\{(\typeN,\emptyset)\}$
\end{tabbing}
\end{algorithm}

\begin{algorithm}
\caption{Function \textbf{TestRMIS}$(\CT,x,IN,OUT)$}\label{algo:testrmis}
\textbf{Parameters:} An \ABCT \CT, a node $x\in \C$ of \CT and two subsets $IN$ and $OUT$ of \C\\
\textbf{Return:} $\bot$ or a subset of the set of nodes of the subgraph of $G$ induced by $x$
\begin{tabbing}
XX: \= XX \= XX \= XX \= XX \= XX \= XX \= XX \=\kill
01: \> Let $C$ be the subgraph of $G$ induced by $x$\\
02: \> Let $R$ be the set of edges $\{u,v\}$ of $C$ such that $L(u)$ and $L(v)$ contain both\\
\> \> \> a label of type \typePO\\
03: \> \textbf{If} $C\setminus R$ is not bipartite \textbf{then}\\
04: \> \> Return $\bot$\\
05: \> Let $C_1,\ldots,C_k$ be the maximal connected components of $C\setminus R$\\
06: \> Let $F$ be an empty $2$-SAT expression on the set of boolean variables $\{x_1,\ldots,x_k\}$\\
07: \> \textbf{Foreach} $i\in\{1,\ldots,k\}$ \textbf{do}\\
08: \> \> Label each node $v\in C_i$ with $\ell(v)=x_i$ or $\ell(v)=\neg x_i$ such that two neighbors in $C_i$\\
\> \> \> \> do not have the same label\\
09: \> \textbf{Foreach} $a\in\A$ such that $a\in C$ \textbf{do}\\
10: \> \> \textbf{If} $L(a)$ is reduced to one label of type \typePI \textbf{then}\\
11: \> \> \> $F:=F\wedge(\ell(a))$\\
12: \> \> \textbf{If} $L(a)$ is reduced to one label of type \typePO or \typePE \textbf{then}\\
13: \> \> \> $F:=F\wedge(\neg\ell(a))$\\
14: \> \textbf{Foreach} $\{u,v\}\in R$ \textbf{do}\\
15: \> \> $F:=F\wedge(\neg\ell(u)\vee\neg\ell(v))$\\
16: \> \textbf{Foreach} $v\in IN$ \textbf{do}\\
17: \> \> $F:=F\wedge(\ell(v))$\\
18: \> \textbf{Foreach} $v\in OUT$ \textbf{do}\\
19: \> \> $F:=F\wedge(\neg\ell(v))$\\
20: \> \textbf{If} $F$ does not have an satisfying assignment \textbf{then}\\
21: \> \> Return $\bot$\\
22: \> \textbf{Else}\\
23: \> \> Let $SA$ be an satisfying assignment of $F$\\
24: \> \> $R:=\{v\in C|\ell(v)=true \text{ in } SA \}$\\
25: \> \> Return $\left(\bigcup_{c\in C(x)}Set(R,c)\right)\cup R$ where\\
\> \> \> \> $Set(R,c)$ is the set of the label of type \typePI of $L(c)$ if $c\in R$,\\
\> \> \> \> the set of the label of type \typePO or \typePE of $L(c)$ if $c\notin R$
\end{tabbing}
\end{algorithm}

\newpage
~
\newpage
~
\newpage

\section{Proof of Algorithm of Section \ref{sec:existsMIS}}\label{sec:bigproof}

We consider that we apply \textbf{FindRMIS} to a graph $G=(V,E)$ whose \ABCT is $\CT=(\A\cup\B\cup\C\cup\PV,E')$. The objective of this section is to prove that \textbf{FindRMIS} terminates in a polynomial time and returns a robust MIS of $G$ if this latter admits one, $\emptyset$ otherwise. Our proof contains mainly three steps. First (Section \ref{sub:prooftree}), we prove that \textbf{FindRMIS} returns a robust MIS of $G$ whenever $G$ is a tree. Once this trivial case eliminated, we define a set of notations and definitions in Section \ref{sub:proofnotations}. These notations are used in the sequel of the proof. Then, we prove central properties provided by Function \textbf{TestRMIS} (Section \ref{sub:testrmis}) and by Functions \textbf{LabelNode} (Section \ref{sub:labelnode}). We use these properties to prove that the execution of \textbf{FindRMIS} up to Line 19 produces a well-labeling of \CT (in a sense defined below) in Section \ref{sub:prooflabelling}. Finally, we prove that \textbf{FindRMIS} uses this well-labelling of \CT to return a robust MIS if $G$ admits one, $\emptyset$ otherwise (Section \ref{sub:proofbuilding}).

\subsection{Case of the tree}\label{sub:prooftree}

\begin{lemma}\label{lem:findrmistree}
If $G$ is a tree, \textbf{FindRMIS} terminates in polynomial time and returns a robust MIS of $G$.
\end{lemma}

\begin{proof}
Assume that $G$ is a tree. Line 01 of \textbf{FindRMIS} build the \ABCT of $G$ (that take a polynomial time in the size of $G$). Then, the test on Line 02 is true (since any biconnected component of a tree has a size of 1) and \textbf{FindRMIS} returns a non empty set of nodes that share the same color in a 2-coloring of $G$ (computed in a polynomial time) on Line 04. Note that this set is trivially a MIS of the tree $G$ and hence a robust MIS of $G$ since any tree belongs to \forallMIS (see Theorem \ref{th:forallMIS}).
\end{proof}

\subsection{Notations and Definitions}\label{sub:proofnotations}

In the following of the proof, according to Lemma \ref{lem:findrmistree}, we restrict our attention to the case where the graph analysed by \textbf{FindRMIS} is not a tree. We define in the following a set of notations used in the proof.

The \ABCT of $G$ is now rooted towards a node $r\in\C$. First, we denote by $\CT(x)=(\A_x\cup\B_x\cup\C_x\cup\PV_x,E'_x)$ the \ABCST of \CT rooted to a node $x\in\A\cup\B\cup\C\cup\PV$. For any node $x\in\A\cup\B\cup\C\cup\PV$, we denote its set of children in \CT by $C(x)$. For any node $x\in\A\cup\B\cup\C\cup\PV\setminus\{r\}$, we denote its parent in \CT by $P(x)$ and its attachment point ($x$ if $x\in\A\cup\PV$, $P(x)$ otherwise) by $AP(x)$.

For any node $x\in\A\cup\B\cup\C\cup\PV$, we say that $x$ induces the subgraph $G(x)=(V_x,E_x)$ of $G$ defined as follows. If $x\in\A\cup\PV$, then $V_x=\{x\}$ and $E_x=\emptyset$. If $x=\{u,v\}\in\B$, then $V_x=\{u,v\}$ and $E_x=\{\{u,v\}\}$. If $x\in\C$, then $V_x=x$ and $E_x=\{\{u,v\}\in E|u\in V_x,v\in V_x\}$. Then, we define the subgraph of $G$ induced by $\CT(x)$ as $G(\CT(x))=(\bigcup_{x\in \A_x\cup\B_x\cup\C_x\cup\PV_x}(V_x), \bigcup_{x\in \A_x\cup\B_x\cup\C_x\cup\PV_x}(E_x))$ and the aerial subgraph of $G$ induced by $\CT(x)$ as $G_a(\CT(x))=(\bigcup_{x\in \A_x\cup\B_x\cup\C_x\cup\PV_x}(V_x)\cup\{a\},$ $\bigcup_{x\in \A_x\cup\B_x\cup\C_x\cup\PV_x}(E_x)\cup\{\{a,AP(x)\}\})$ with $a\notin \A_x\cup\B_x\cup\C_x\cup\PV_x$. 


The algorithm \textbf{FindRMIS} associates a set of labels $L(x)$ to each node $x\in\A\cup\B\cup\C\cup\PV$. A label is a couple $(t,S)$ with type $t\in \{\typePI,\typePO,\typePE,\typeN,\typeE\}$ and $S\subseteq V$.

We are now in measure to introduce the main definitions on which relies our proof.

\begin{definition}[Well-labeled \ABCST]\label{def:welllabeledabcst}
Given a node $v\in\A\cup\B\cup\C\cup\PV\setminus\{r\}$, the \ABCST $\CT(v)$ is well-labeled if the following properties hold for any node $x\in \A_x\cup\B_x\cup\C_x\cup\PV_x$:
\begin{enumerate}
\item $L(x)=\{(\typeN,\emptyset)\}$ if and only if $G(\CT(x))$ does not admit a robust MIS and $G_a(\CT(x))$ does not admit a robust MIS including $a$
\item $L(x)$ contains $(\typePI,M)$ with $M$ a robust MIS of $G(\CT(x))$ including $AP(x)$ if and only if $G(\CT(x))$ admits such a MIS.
\item $L(x)$ contains $(\typePO,M)$ with $M$ a robust MIS of $G(\CT(x))$ not including $AP(x)$ if and only if $G(\CT(x))$ admits such a MIS.
\item $L(x)$ contains $(\typePE,M\setminus\{a\})$ with $M$ a robust MIS of $G_a(\CT(x))$ including $a$ if and only if $G_a(\CT(x))$ admits such a MIS and $G(\CT(x))$ does not admit a robust MIS not including $AP(x)$.
\end{enumerate}
\end{definition}

\begin{definition}[Well-labeled \ABCT]\label{def:welllabeledabct}
The \ABCT $\CT$ is well-labeled if the following properties hold:
\begin{enumerate}
\item $L(r)=\{(\typeE,M)\}$ with $M$ a robust MIS of $G$ if and only if $G$ admits such a MIS.
\item $L(r)=\{(\typeN,\emptyset)\}$ if and only if $G$ does not admit a robust MIS.
\end{enumerate}
\end{definition}

\subsection{Function TestRMIS}\label{sub:testrmis}

In this section, we prove the main technical part of the algorithm. Roughly speaking, we prove that the function \textbf{TestRMIS} applied to any node of $\C$ is able to determine if the \ABCST rooted to this node admits a robust MIS or not and to compute one such MIS. We need four lemmas depending on the parameters of the function.

\begin{lemma}\label{lem:testrmis1}
For any node $x\in\C\setminus\{r\}$ such that, for each $v$ of $C(x)$, $\CT(v)$ is well-labeled, $L(v)$ does not contain a label of type \typeN, and $L(P(x))=\emptyset$, \textbf{TestRMIS}$(\CT,x,\{P(x)\},\emptyset)$ returns (in polynomial time):
\begin{itemize}
\item $\bot$ if $G(\CT(x))$ does not admit a robust MIS including $P(x)$;
\item $M$ otherwise (with $M$ such a MIS).
\end{itemize}
\end{lemma}

\begin{proof}
Let $x$ be a node of $\C\setminus\{r\}$ such that, for each $v$ of $C(x)$, $\CT(v)$ is well-labeled, $L(v)$ does not contain a label of type \typeN, and $L(P(x))=\emptyset$.

First, assume that \textbf{TestRMIS}$(\CT,x,\{P(x)\},\emptyset)$ is executed when $G(\CT(x))$ does not admit a robust MIS including $P(x)$. Then, we are going to prove that, if the test of Line 03 of \textbf{TestRMIS} is false, the one of Line 20 of \textbf{TestRMIS} is true (and hence that \textbf{TestRMIS} returns necessarily $\bot$ in this case). By contradiction, assume that $F$ (build up to Line 19 of \textbf{TestRMIS}) admits a satisfying assignment $SA$.

As, for each $v$ of $C(x)$, $\CT(v)$ is well-labeled and $L(v)$ does not contain a label of type \typeN, each $G(\CT(v))$ admits a robust MIS or $G_a(\CT(v))$ admits a robust MIS including $a$ by definition. That allows us to define the following sets:
\begin{itemize}
\item $R=\{v\in C|\ell(v)=true \text{ in } SA \}$;
\item For any $c\in C(x)$ such that $\ell(c)=true$ in $SA$, $M_c$ is a robust MIS of $G(\CT(c))$ such that $c\in M_c$;
\item For any $c\in C(x)$ such that $\ell(c)=false$ in $SA$, $M_c$ is a robust MIS of $G(\CT(c))$ such that $c\notin M_c$ if such a MIS exists, $M_c\cup\{a\}$ is a robust MIS of $G_a(\CT(c))$ otherwise.
\end{itemize}
Then, we are going to prove that the set $M=R\bigcup(\{M_c|c\in C(x)\})$ is a robust MIS of $G(\CT(x))$. Note that $R=M|_x$ (by construction) and $P(x)\in M$ (thanks to the clause introduced in $F$ on Line 17 of \textbf{TestRMIS}). 

\begin{description}
\item[Independence of $M$:] As each $M_c$ is independent and covers the articulation point that connects $G(\CT(c))$ to $C$ for each $c\in C(x)$ by construction, it remains to prove that $R$ is independent. Let $e=\{u,v\}$ be an edge of $C$. If $e\in R$, then the clause introduced in $F$ on Line 15 of \textbf{TestRMIS} ensures that labels of $u$ and $v$ cannot be simultaneously true in $SA$. Otherwise, $e$ belongs to $C\setminus R$ (that is bipartite) and hence, the labeling done on Line 08 of \textbf{TestRMIS} ensures us that labels of $u$ and $v$ cannot be simultaneously true in $SA$. Then, the construction of $R$ guarantees its independence.

\item[Maximality of $M$:] As each $M_c$ is maximal (in $G(\CT(c))$ or $G_a(\CT(c))$ depending on the case) for each $c\in C(x)$ by construction, it remains to prove that $R$ is maximal in $C$. Let $u\in C$ be a node that does not belong to $M$. That implies that $\ell(u)$ is false in $SA$ and then that $M_u$ is a robust MIS of $G(\CT(u))$ such that $u\notin M_u$ if such a MIS exists, a robust MIS of $G_a(\CT(u))$ such that $a\in M_u$ otherwise. 

In the first case, $u$ has a neighbor in $M_u$ by maximality of $M_u$. In the second case, as $\CT(u)$ is well-labeled by assumption, we know that $L(u)$ contains a label of type \typePE and hence does not contains a label of type \typePO. Then, no adjacent edge to $u$ in $C$ belongs to $R$. As a consequence, $u$ and its neighbors belongs to the same connected component of $C\setminus R$ and receive opposite labels on Line 08 of \textbf{TestRMIS}. In both cases, $u$ has a neighbor in $M$, that proves its maximality.

\item[Robustness of $M$:] The robustness of $M$ is proved by using the following equivalence (proved by \cite{DKP15} in the case of minimal dominating set but easily translatable to MIS): a MIS $M$ of a graph $G$ is robust if and only if, for any node $u$ not in $M$, removing all edges between $u$ and a node of $M$ disconnects $G$. Let $u$ be a node of $G(\CT(x))$ that does not belong to $M$. 

Assume first that $u\in C$. If $u$ has no adjacent edge in the set $R$ defined in Line 02 of \textbf{TestRMIS}, then all its neighbors in $C$ belongs to $M$ by the labeling done on Line 08 of \textbf{TestRMIS}. By definition of a biconnected component, the removing of all edges between $u$ and its neighbors in $C$ deconnects $G$. Otherwise (\ie $u$ has at least one adjacent edge in the set $R$), that means that $L(u)$ contains a label of type \typePO. As $\CT(u)$ is well-labeled by assumption, we know that $M_u$ is a robust MIS of $G(\CT(u))$ such that $u\notin M_u$. By robustness of $M_u$ on $G(\CT(u))$, we know that the removing of all edges between $u$ and its neighbors of $G(\CT(u))$ in $M_u$ deconnects $G(\CT(u))$ (hence $G$).

Assume now that $u\notin C$. The result is easily proved by the robustness of $M_c$ where $c$ is the child $x$ such that $u\in\CT(c)$.
\end{description}

The set $M$ is hence a robust MIS of $G(\CT(x))$ such that $P(x)\in M$, that contradicts the assumption that $G(\CT(x))$ does not admit such a MIS. In conclusion, \textbf{TestRMIS}$(\CT,x,\{P(x)\},\emptyset)$ returns $\bot$ if $G(\CT(x))$ that does not admit a robust MIS $M$ such that $P(x)\in M$.

Second, assume that $G(\CT(x))$ admits a robust MIS $M$ such that $P(x)\in M$. If two neighbors $u$ and $v$ of $C$ do not belongs to $M$, then the edge $\{u,v\}$ belongs to the set $R$ defined in Line 02 of \textbf{TestRMIS} (otherwise, we obtain a contradiction with the maximality of $M$). Then, the subgraph $C\setminus R$ is bipartite (one partition is $M|_x$, the other is $C\setminus M|_x$) and hence \textbf{TestRMIS} does not return $\bot$ on Line 04. As we can deduce a satisfying assignment to the 2-SAT formula (build up to Line 19) from $M$ (it is sufficient to assign all $x_i$ ---$i\in\{1,\ldots,k\}$--- to have $\forall v\in C, \ell(v)=true\Leftrightarrow v\in M$), we can deduce that \textbf{TestRMIS} does not return $\bot$ on Line 21. In conclusion, \textbf{TestRMIS} returns a set $M$ on Line 25. We know, by the construction of this set and by the proof of the first case, that $M$ is a robust MIS of $G(\CT(x))$ such that $P(x)\in M$.

To conclude the proof, note that all instructions of \textbf{TestRMIS} are polynomial in the size of $C$ and are repeated at most $|C|$ times, that implies that the running time of \textbf{TestRMIS} is polynomial in the size of $G$.
\end{proof}

The three following lemmas shows similar results depending on the parameters of the Function \textbf{TestRMIS}. Their proofs are similar to the previous one.

\begin{lemma}\label{lem:testrmis2}
For any node $x\in\C\setminus\{r\}$ such that, for each $v$ of $C(x)$, $\CT(v)$ is well-labeled, $L(v)$ does not contain a label of type \typeN, and $L(P(x))=\emptyset$, \textbf{TestRMIS}$(\CT,x,\emptyset,\{P(x)\})$ returns (in polynomial time):
\begin{itemize}
\item $\bot$ if $G(\CT(x))$ does not admit a robust MIS not including $P(x)$;
\item $M$ otherwise (with $M$ such a MIS).
\end{itemize}
\end{lemma}

\begin{lemma}\label{lem:testrmis3}
For any node $x\in\C\setminus\{r\}$ such that, for each $v$ of $C(x)$, $\CT(v)$ is well-labeled, $L(v)$ does not contain a label of type \typeN, and $L(P(x))=\{(\typePO,\emptyset)\}$, \textbf{TestRMIS}$(\CT,x,\emptyset,\emptyset)$ returns (in polynomial time):
\begin{itemize}
\item $\bot$ if $G_a(\CT(x))$ does not admit a robust MIS including $a$;
\item $M\setminus\{a\}$ otherwise (with $M$ such a MIS).
\end{itemize}
\end{lemma}

\begin{lemma}\label{lem:testrmis4}
If, for each $v$ of $C(r)$, $\CT(v)$ is well-labeled, $L(v)$ does not contain a label of type \typeN, \textbf{TestRMIS}$(\CT,r,\emptyset,\emptyset)$ returns (in polynomial time):
\begin{itemize}
\item $\bot$ if $G$ does not admit a robust MIS;
\item $M$ otherwise (with $M$ such a MIS).
\end{itemize}
\end{lemma}

\subsection{Functions LabelNode}\label{sub:labelnode}

In this section, we prove that each of the three functions \textbf{LabelNode} produces a well-labeled \ABCST rooted on a node $x$ provided that all \ABCST rooted on children of $x$ are well-labeled. 

\begin{lemma}\label{lem:articulation}
For any node $x\in\A$ such that, for each $v$ of $C(x)$, $\CT(v)$ is well-labeled and $L(v)$ does not contain a label of type \typeN, $\CT(x)$ is well-labeled after the execution (in polynomial time) of \textbf{LabelNodeA}$(\CT,x)$.
\end{lemma}

\begin{proof}
Observe that, for any node $x\in\A$, we have the following properties by definition of an articulation point:
\begin{itemize}
\item $G(\CT(x))$ does not admit a robust MIS and $G_a(\CT(x))$ does not admit a robust MIS including $a$ if and only if there exists at least one child $c\in C(x)$ such that $G(\CT(c))$ does not admit a robust MIS and $G_a(\CT(c))$ does not admit a robust MIS including $a$.
\item $G(\CT(x))$ admits a robust MIS including $AP(x)$ if and only if, for every child $c\in C(x)$, $G(\CT(c))$ admits a robust MIS $M_c$ including $AP(c)=x$. Moreover, $M=\bigcup_{c\in C(x)}M_c$ is such a MIS of $G(\CT(x))$.
\item $G(\CT(x))$ does not admit a robust MIS not including $AP(x)=x$ and $G_a(\CT(x))$ admits a robust MIS including $a\in M$ if and only if, for every child $c\in C(x)$, $G(\CT(c))$ does not admit a robust MIS not including $AP(c)=x$ and $G_a(\CT(c))$ admits a robust MIS $M_c$ including $a$. Moreover, $M=\bigcup_{c\in C(x)}M_c$ is such a MIS of $G_a(\CT(x))$.
\item $G(\CT(x))$ admits a robust MIS not including $AP(x)=x$ if and only, for every child $c\in C(x)$, $G(\CT(c))$ admits a robust MIS $M_c$ not including $AP(c)=x$ or $G_a(\CT(c))$ admits a robust MIS including $a$ and there exists at least one child $c\in C(x)$ such that $G(\CT(c))$ admits a robust MIS $M_c$ including $AP(c)=x$. Moreover, $M=\bigcup_{c\in C(x)}M_c$ is such a MIS of $G(\CT(x))$.

\end{itemize}

Let $x$ be a node of $\A$ such that, for each $v$ of $C(x)$, $\CT(v)$ is well-labeled and $L(v)$ does not contain a label of type \typeN. Then, note that the construction of \textbf{LabelNodeA} is strictly based on the previous set of properties, implying that $\CT(x)$ is well-labeled after the execution of \textbf{LabelNodeA}$(\CT,x)$.

To conclude the proof, note that all tests of \textbf{LabelNodeA} are performed in polynomial time (since the size of the set of children is bounded by the size of $G$), hence \textbf{LabelNodeA} terminates in a polynomial time.
\end{proof}

\begin{lemma}\label{lem:bridge}
For any node $x\in\B$ with $C(x)=\{v\}$ such that $\CT(v)$ is well-labeled and $L(v)$ does not contain a label of type \typeN, $\CT(x)$ is well-labeled after the execution (in constant time) of \textbf{LabelNodeB}$(\CT,x)$.
\end{lemma}

\begin{proof}
Observe that, for any node $x\in\B$ with $C(x)=\{v\}$, we have the following properties by definition of a bridge:
\begin{itemize}
\item $G(\CT(x))$ does not admit a robust MIS and $G_a(\CT(x))$ does not admit a robust MIS including $a$ if and only if $G(\CT(v))$ does not admit a robust MIS and $G_a(\CT(v))$ does not admit a robust MIS including $a$.

\item $G(\CT(x))$ admits a robust MIS including $AP(x)$ if and only $G(\CT(v))$ does not admit a robust MIS including $AP(v)=v$ and $G_a(\CT(v))$ admits a robust MIS $M_a$ including $a$ or $G(\CT(v))$ admits a robust MIS $M_v$ not including $AP(v)=v$. Moreover, $(M_a\setminus\{a\})\cup\{AP(x)\}$ and $M_v\cup\{AP(x)\}$ are respectively such a MIS of $G(\CT(x))$.

\item $G(\CT(x))$ admits a robust MIS not including $AP(x)$ if and only $G(\CT(v))$ admits a robust MIS $M_v$ including $AP(v)=v$. Moreover, $M_v$ is such a MIS of $G(\CT(x))$.

\item $G(\CT(x))$ does not admit a robust MIS not including $AP(x)$ and $G_a(\CT(x))$ admits a robust MIS including $a$ if and only if $G(\CT(v))$ admits a robust MIS $M_v$ not including $AP(v)$. Moreover, $M_v$ is such a MIS of $G(\CT(x))$.

\end{itemize}

Let $x$ be a node of $\B$ with $C(x)=\{v\}$ such that $\CT(v)$ is well-labeled and $L(v)$ does not contain a label of type \typeN. Then, note that the construction of \textbf{LabelNodeB} is strictly based on the previous set of properties, implying that $\CT(x)$ is well-labeled after the execution of \textbf{LabelNodeB}$(\CT,x)$.

To conclude the proof, note that all tests of \textbf{LabelNodeB} are performed in constant time, hence \textbf{LabelNodeB} terminates in a constant time.
\end{proof}

\begin{lemma}\label{lem:component}
For any node $x\in\C\setminus\{r\}$ such that, for each $v$ of $C(x)$, $\CT(v)$ is well-labeled and $L(v)$ does not contain a label of type \typeN, $\CT(x)$ is well-labeled after the execution (in polynomial time) of \textbf{LabelNodeC}$(\CT,x)$.
\end{lemma}

\begin{proof}
This result directly follows from the definition of a well-labeled \ABCST, the construction of the function \textbf{LabelNodeC}, and Lemmas \ref{lem:testrmis1} to \ref{lem:testrmis4}.
\end{proof}

\subsection{Labelling of the \ABCT}\label{sub:prooflabelling}

We are now in measure to characterize the labeling made by Functions \textbf{LabelSubTree} and \textbf{FindRMIS} (up to Line 19).

\begin{lemma}\label{lem:subtree}
For any node $x\in\A\cup\B\cup\C\cup\PV\setminus\{r\}$, $\CT(x)$ is well-labeled after the execution (in polynomial time) of \textbf{LabelSubTree}$(\CT,x)$.
\end{lemma}

\begin{proof}
We prove this result by induction on $h_x$, the height of the \ABCST $\CT(x)$.

Consider first a node $x\in\A\cup\B\cup\C\cup\PV\setminus\{r\}$ such that $h_x=0$. That implies that $x\in\C\cup\PV\setminus\{r\}$ and that $C(x)=\emptyset$. If $x\in\C\setminus\{r\}$, the result holds by Lemma \ref{lem:component}. If $x\in\PV$, Line 13 of \textbf{LabelSubTree} is executed and produces a well-labeling of $x$ by definition of a pendant node. Note that \textbf{LabelSubTree} terminates in polynomial time in both cases.

Consider now a node $x\in\A\cup\B\cup\C\setminus\{r\}$ such that $h_x>0$. The execution of \textbf{LabelSubTree} starts by a recursive call on each child of $x$. By induction assumption, each \ABCST $\CT(c)$ is well-labeled in polynomial time with $c\in C(x)$.

As $G(\CT(x))$ does not admit a robust MIS and $G_a(\CT(x))$ does not admit a robust MIS $M$ such that $a\in M$ if there exists one child $c\in C(x)$ such that $G(\CT(c))$ does not admit a robust MIS and $G_a(\CT(c))$ does not admit a robust MIS $M_c$ such that $a\in M_c$, the well-labeling of $\CT(c)$ and the test on Lines 03-04 of \textbf{LabelSubTree} imply that $x$ is labeled $\{(\typeN,\emptyset\}$ (in polynomial time) in this case.

Then, Lemmas \ref{lem:articulation}, \ref{lem:bridge}, and \ref{lem:component} allow us to conclude that the end of the execution of \textbf{LabelSubTree} (Lines 06 to 13) well-labels $\CT(x)$ in polynomial time.
\end{proof}

\begin{lemma}\label{lem:labeling}
If $G$ is not a tree, the execution of \textbf{FindRMIS} up to Line 19 produces a well-labeled \ABCT of $G$ in polynomial time.
\end{lemma}

\begin{proof}
By Lemmas \ref{lem:findrmistree} and \ref{lem:subtree}, the execution of \textbf{FindRMIS} on a graph $G$ (that is not a tree) up to Line 11 leads in polynomial time to a well-labeling of each \ABCST $\CT(c)$ with $c\in C(x)$.

As $G$ does not admit a robust MIS if there exists one child $c\in C(x)$ such that $G(\CT(c))$ does not admit a robust MIS and $G_a(\CT(c))$ does not admit a robust MIS $M_c$ such that $a\in M_c$, the well-labeling of $\CT(c)$ and the test on Lines 12-13 of \textbf{FindRMIS} imply that $r$ is labeled $\{(\typeN,\emptyset\}$ (in polynomial time) in this case.

Then, Lemma \ref{lem:testrmis4} allows us to conclude that the execution of \textbf{FindRMIS} up to Line 19 well-labels $\CT$ in polynomial time.
\end{proof}

\subsection{The end of the road}\label{sub:proofbuilding}

\begin{theorem}
The execution of of \textbf{FindRMIS} on a graph $G$ returns in polynomial time a robust MIS of $G$ if $G$ admits one, $\emptyset$ otherwise.
\end{theorem}

\begin{proof}
If $G$ is a tree, the result directly follows from Lemma \ref{lem:findrmistree}. Otherwise, Lemma \ref{lem:labeling} imply that the execution of \textbf{FindRMIS} up to Line 19 produces a well-labeled \ABCT of $G$ in polynomial time. 

Then, by definition of a well-labeled \ABCT, $L(r)=\{(\typeN,\emptyset)\}$ if and only if $G$ does not admit a robust MIS. Then, \textbf{FindRMIS} returns $\emptyset$ (in polynomial time) on Line 21 if and only if $G$ does not admit a robust MIS.

Hence, in the case where $G$ admits at least one robust MIS, \textbf{FindRMIS} terminates (in polynomial time) on Line 23. By definition of a well-labeled \ABCT, the returned set is a robust MIS of $G$, that ends the proof.
\end{proof}

\newpage

\section{Extra Materials}
\label{sec:extra}

Figure~\ref{fig:tagged} shows how the \ABCT in Figure~\ref{fig:ex_compo} is tagged after completion of the
labelling phase of Algorithm~\textbf{FindRMIS}---see Algorithm~\ref{algo:findrmis} in Annexe~\ref{sec:pseudo}. 

\begin{figure}[htb]
\begin{center}
\includegraphics[width=.7\textwidth] {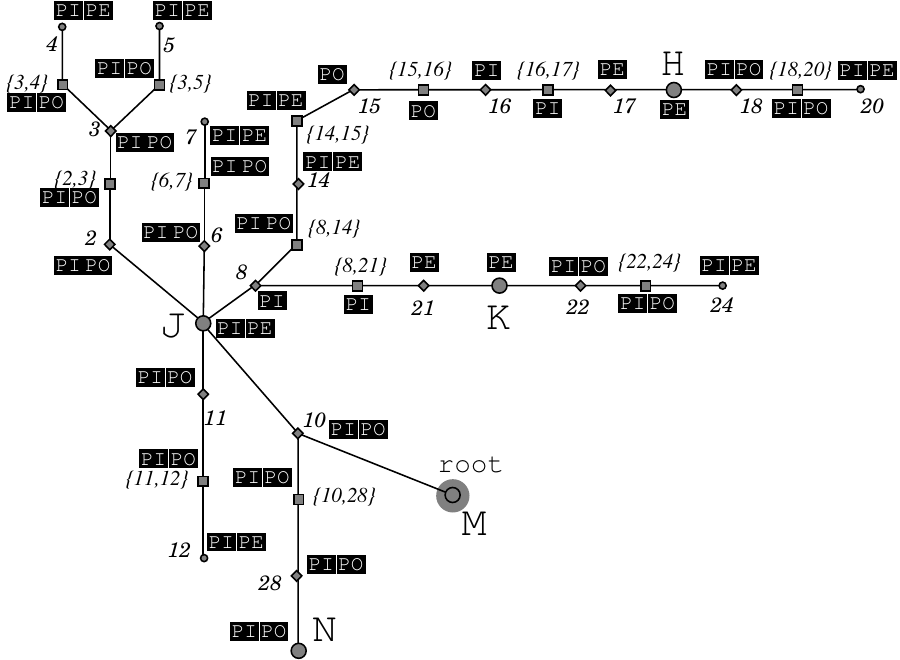}
\vspace{-0.5cm}
\end{center}
\caption{The \ABCT in Figure~\ref{fig:ex_compo} tagged after completion of the labelling phase.}\label{fig:tagged}
\end{figure}

The graph shown in Figure~\ref{fig:ex_compo} admits two robust MISs shown in Figures~\ref{fig:ex_rmis1}
and~\ref{fig:ex_rmis2}.

\begin{figure}[htb]
 \begin{minipage}[b]{.45\linewidth}
  \centering \includegraphics[width=\textwidth] {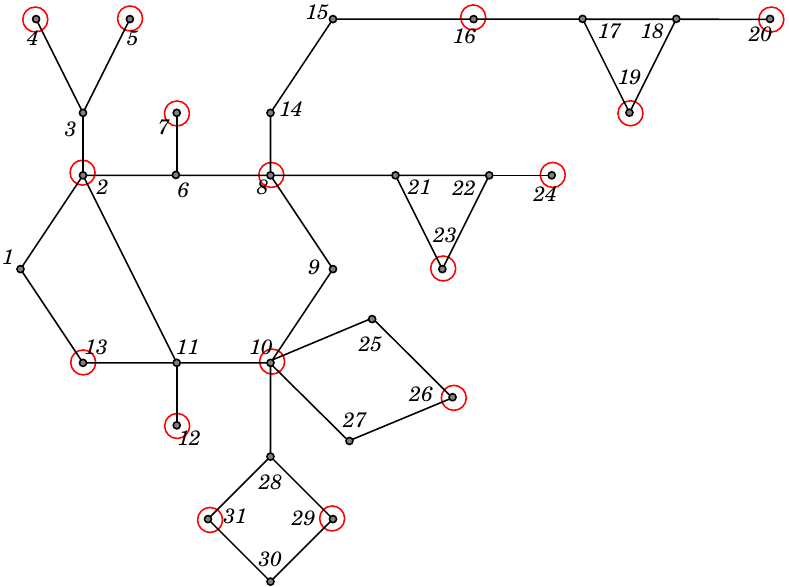}
 \caption{One possible RMIS of $G$ that includes $16$ vertices.  \label{fig:ex_rmis1}}
\end{minipage} \hfill
 \begin{minipage}[b]{.45\linewidth}
  \centering \includegraphics[width=\textwidth] {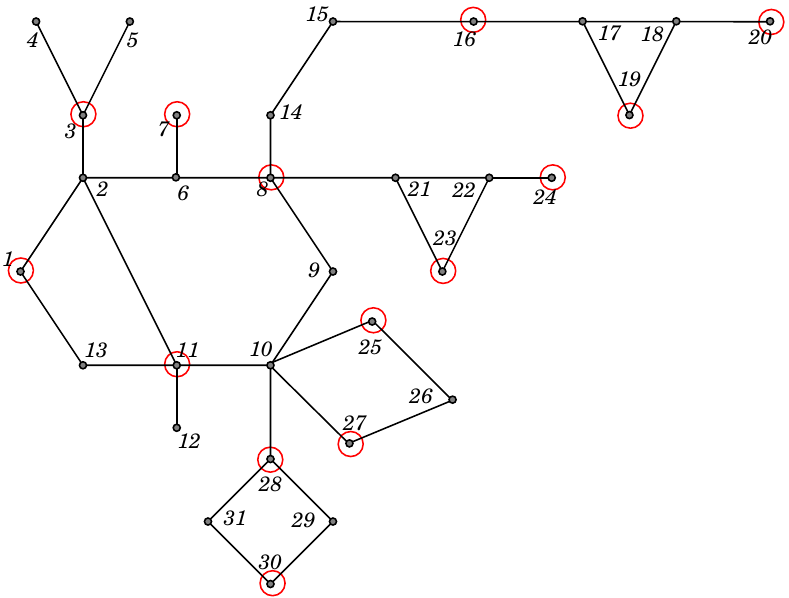}
 \caption{Another possible RMIS of $G$ that includes $14$ vertices. \label{fig:ex_rmis2}}
 \end{minipage}\hfill
\end{figure}

\end{document}